\documentclass[11p,reqno]{amsart}

\usepackage[T1]{fontenc}
\usepackage{graphicx}
\usepackage{amssymb,amsthm,amsmath,mathrsfs,bm,braket,marginnote}
\usepackage{enumerate}
\usepackage{appendix}
\usepackage[colorlinks=true, pdfstartview=FitV, linkcolor=blue, citecolor=blue, urlcolor=blue]{hyperref}

\usepackage{pgf}
\usepackage{pgfplots}
\usepackage{tikz}
\usetikzlibrary{arrows,calc}
\usepackage{verbatim}
\usetikzlibrary{decorations.pathreplacing,decorations.pathmorphing}
\usepackage{multicol}  
\usepackage{multirow}

\numberwithin{equation}{section}
\theoremstyle{plain}
\newtheorem{theorem}{Theorem}[section]
\newtheorem{lemma}[theorem]{Lemma}
\newtheorem{coro}[theorem]{Corollary}

\theoremstyle{definition}
\newtheorem{define}[theorem]{Definition}
\newtheorem{remark}[theorem]{Remark}

 \newcommand{\sgn}[1]{\operatorname{sgn} #1}

 \reversemarginpar
\begin{document}

\title[DP peakon and C-Toda lattices]{Degasperis-Procesi peakon dynamical system and finite Toda lattice of CKP type} 
\author{Xiang-Ke Chang}
\address{ LSEC, ICMSEC, Academy of Mathematics and Systems Science, Chinese Academy of Sciences, P.O.Box 2719, Beijing 100190, PR China; and School of Mathematical Sciences, University of Chinese Academy of Sciences, Beijing 100049, PR China}
\email{changxk@lsec.cc.ac.cn}

\author{Xing-Biao Hu}
\address{LSEC, ICMSEC, Academy of Mathematics and Systems Science, Chinese Academy of Sciences, P.O.Box 2719, Beijing 100190, PR China; and School of Mathematical Sciences, University of Chinese Academy of Sciences, Beijing 100049, PR China}
\email{hxb@lsec.cc.ac.cn}

\author{Shi-Hao Li}
\address{LSEC, ICMSEC, Academy of Mathematics and Systems Science, Chinese Academy of Sciences, P.O.Box 2719, Beijing 100190, PR China; and School of Mathematical Sciences, University of Chinese Academy of Sciences, Beijing 100049, PR China}
\email{lishihao@lsec.cc.ac.cn}

\date{}

\dedicatory{}


\begin{abstract} 
In this paper, we propose a finite Toda lattice of CKP type (C-Toda) together with a Lax pair. Our motivation is based on the fact that the Camassa-Holm (CH) peakon dynamical system and the finite Toda lattice may be regarded as opposite flows in some sense. As an intriguing analogue to the CH equation, the Degasperis-Procesi (DP) equation also supports the presence of peakon solutions. Noticing that the peakon solution to the DP equation is expressed in terms of bimoment determinants related to the Cauchy kernel, we impose opposite time evolution on the moments and derive the corresponding bilinear equation.  {The corresponding quartic representation is shown to be a continuum limit of a discrete CKP equation, due to which we call the obtained equation finite Toda lattice of CKP type. Then, a nonlinear version of the C-Toda lattice together with a Lax pair is derived. }As a result,   
it is shown that the DP peakon lattice and the finite C-Toda lattice form opposite flows under certain transformation.


%
\end{abstract}

\maketitle
\section{Introduction}
The celebrated Toda lattice and the Camassa-Holm (CH) equation have both attracted great attention in the course of the development on integrable systems. The Toda lattice was originally introduced by Toda \cite{toda1967vibration} as a simple model for describing a chain of particles with nearest neighbor exponential interaction and later it was frequently studied under the Flaschka's form \cite{date1976analogue,dubrovin1976non,flaschka1974toda1,moser1975finitely}. The CH equation, named after Camassa and Holm \cite{camassa1993integrable}, arises as a shallow water wave model and owns many interesting properties  \cite{beals2000multipeakons,constantin1998wave,constantin2000stability,holden2008global,lenells2004stability}. One of the most attractive characters is that it admits a kind of soliton solutions with peaks (or called peakons), whose dynamics could be described by a system of ODEs. These special solutions 
seem to capture main attributes of solutions of the CH equation: the breakdown of regularity which can be interpreted as collisions of peakons, and the nature of 
long time asymptotics which can be loosely described as peakons becoming 
free particles in the asymptotic region \cite{beals2000multipeakons}. Interestingly, there exists certain intimate connection between the CH peakon lattice (i.e. the ODE system describing the CH peakons)  and the finite Toda lattice \cite{beals2001peakons,ragnisco1996peakons}.


The finite Toda lattice, i.e. the case of a chain of finitely many particles ($x_0=-\infty, x_{n+1}=\infty$), was investigated by Moser \cite{moser1975finitely}, who employed inverse spectral method to analyze the spectral and inverse spectral problems related to Jacobi matrix. It was shown by Beals, Sattinger and Szmigielski \cite{beals2000multipeakons} that, the spectral problem of the CH peakon lattice is related to a finite discrete string problem, which can also be solved by use of inverse spectral method.  The explicit formulae of the solutions can both be expressed in terms of Hankel determinants with moments of discrete measures. In a follow-up work \cite{beals2001peakons}, Beals, Sattinger and Szmigielski indicated that there exists a bijective map from a discrete string problem with positive weights to Jacobi matrices, which allows the pure peakon flow of the CH equation to be realized as an isospectral Jacobi flow as well. This gives a unified picture for the Toda and the CH peakon flows. Indeed, this also implies that the CH peakon and Toda lattices can be viewed as opposite flows. Please refer Appendix \ref{app:CHA} for some details.

Peakons have received much attention in the recent two decades. 
In addition to the CH equation,  there subsequently appear many integrable systems with the presence of peakon solutions, among which, the Degasperis-Procesi (DP) equation \cite{degasperis1999asymptotic} and the Novikov equation \cite{hone2008integrable,novikov2009generalisations} are two of the most widely studied equations.  
The dynamics of DP and Novikov peakons are also governed by the respective ODE systems, which can also be explicitly solved by employing inverse spectral method \cite{lundmark2005degasperis,hone2009explicit}. In the search of solutions, a discrete cubic string problem and its dual problem are involved respectively. Eventually, the solutions can be expressed in terms of some bimoment determinants related to Cauchy biorthogonal polynomials. 
A natural question is whether there exist the corresponding opposite flows to DP and Novikov peakon lattices or not.

The answer is positive. In \cite{chang2017application}, the result for the Novikov peakon case is reported. It is shown that the Novikov peakon dynamical system is connected with the finite Toda lattice of BKP type (B-Toda lattice) and their solutions are related to the partition function of Bures ensemble with discrete measure (The readers are invited to see Appendix \ref{app:NVB} for a short summary). This paper is mainly devoted to dealing with the DP peakon case. It turns out that a Toda lattice of CKP type is an opposite flow to the DP peakon lattice. The Toda lattice of CKP type is shown to be Lax integrable and it seems novel to our knowledge. In fact, this equation is of interests in the study of classic integrable system since it can be regarded as a continuum limit of the discrete CKP equation by Bobenko and Schief \cite{bobenko2015discrete,bobenko2017discrete,schief2003lattice}.



According to different types of infinite dimensional Lie groups and the corresponding Lie algebras, classical integrable systems may be classified into  AKP, BKP, CKP types\footnote{Modifications of the Kadomtsev--Petviashvili hierarchy corresponding to Lie algebras of type A, B, C} etc., whose transformation groups are $GL(\infty)$, $O(\infty)$ and $Sp(\infty)$ etc., respectively \cite{date1981kp,jimbo1983solitons}. In fact, the CH equation has a hodograph link  to the first negative flow in the hierarchy of the Korteweg-de Vries (KdV) equation \cite{fuchssteiner1996some}, which belongs to AKP type \cite{hirota2004direct,jimbo1983solitons}.   From the view of the tau-function, the solutions of AKP type admit the closed forms in terms of determinants. It makes sense that determinants appear in the CH peakons.   Besides, there is a hodograph link between the Novikov equation and a negative flow in the Sawada-Kotera (SK) hierarchy (BKP type) \cite{hone2008integrable} and the solutions of BKP type equations are usually expressed not as determinants but as Pfaffians \cite{hirota2004direct}. {Furthermore, the DP equation corresponds to a negative flow in the Kaup-Kupershmidt (KK)  hierachy \cite{degasperis2002new} belonging to CKP type. Due to their respectively intrinsic structures, the challenges we encounter become completely different. The objects we deal with for the DP case here are the complicated Cauchy bimoment determinants, which admit multiple integral representations as the partition function of the Cauchy two-matrix model \cite{bertola2009cauchy,bertola2014cauchy}.


The paper is arranged as follows. In Section \ref{sec:dp}, we give more background on the DP equation and its peakon solutions. And we provide a direct confirmation for the DP peakons by determinant technique so that the DP peakon lattice may be viewed as an isospectral flow on a manifold cut out  by determinant identities. Section \ref{sec:ctoda} is mainly used to search for an opposite flow of the DP peakon lattice. We achieve the goal by use of determinant technique and a novel lattice of Toda type (\ref{eq:ctoda}) is constructed.  Theorem \ref{th:dp_ctoda} summarises the connection between them.{ In Section \ref{sec:schief}, we explain why the obtained lattice is of CKP type by uncovering its connection with a discrete CKP equation studied by Bobenko and Schief \cite{bobenko2015discrete,bobenko2017discrete,schief2003lattice}. } Last of all, the concluding remark is given in Section \ref{sec:con}, where a unified picture is presented for the CH peakon and Toda, Novikov peakon and B-Toda, DP peakon and C-Toda lattices (see Table \ref{comp_abc}).

\section{DP peakons}\label{sec:dp}
In this section, we give more background on the DP equation and its multipeakon solutions, and then reformulate the multipeakons in our own way.

\subsection{On DP equation}
The Degasperis--Procesi (DP) equation
\begin{align}
m_t+(um)_x+2u_xm=0,\qquad m=u-u_{xx}\label{eq:DP}
\end{align}
 was found by Degasperis and Procesi \cite{degasperis1999asymptotic} to pass the necessary (but not sufficient) test of asymptotic integrability, and later shown by Degasperis, Holm, and Hone \cite{degasperis2002new} to be integrable indeed  in sense of Lax pair, bi-Hamiltonian structure, infinitely many conservation laws, etc.

As a modification of the CH equation
\begin{align*}
m_t+(um)_x+u_xm=0,\qquad m=u-u_{xx},
\end{align*}
the DP equation may also be regarded as a model for the propagation of shallow water waves and a rigorous explanation was given in \cite{Constantin-Lannes}. It also admits peakon solutions \cite{lundmark2003multi,lundmark2005degasperis}. In \cite{liu2,liu1}, the authors present important results regarding stability of DP peakons, and \cite{sz2,sz1} deal with collisions of DP peakons. Furthermore, the study on the DP peakon problem induces new questions regarding Nikishin systems \cite{bertola2010cauchy}  studied in approximation theory, and random two-matrix models \cite{bertola2009cauchy,bertola2014cauchy}.  
Despite its superficial similarity to the CH equation, the DP equation has in addition shock solutions \cite{coclite-karlsen-DPwellposedness, coclite-karlsen-DPuniqueness, lundmark-shockpeakons} (see \cite{lundmark-shockpeakons} for the onset of shocks in the form of shockpeakons).  Indeed, there has been considerable interest in the DP equation. For other references, see e.g.  \cite{constantin2012inverse,escher2006global,feng13,feng17,lenells2005traveling,matsuno2005n} etc. (It was not our purpose trying to be exhaustive. Thus, we beg indulgence for the numerous omissions it certainly contains.)
 
 
 In the following,  we will focus on the explicit construction of DP multipeakons in \cite{lundmark2005degasperis}, where the pure multipeakon case was rigorously analyzed by use of inverse spectral technique.


\subsection{Review of the work by Lundmark \& Szmigielski}  
When the multipeakon ansatz 
\begin{align}
u(x,t)=\sum_{j=1}^nm_j(t)e^{-|x-x_j(t)|}
\end{align}
is taken into account, it follows from \eqref{eq:DP} that $m$ can be regarded as a discrete measure
$$
m(x,t)=2\sum_{k=1}^nm_k(t)\delta(x-x_k(t)).
$$
By using distributional calculus, it is known that  the first equation of \eqref{eq:DP} is satisfied in a weak sense if the positions $(x_1,\ldots,x_n)$ and momenta $(m_1, \ldots , m_n)$ of the peakons obey the following system of $2n$ ODEs \cite{degasperis2002new,lundmark2005degasperis}:
\begin{align}
\dot x_k=u(x_k)=\sum_{j=1}^nm_je^{-|x_j-x_k|}, \qquad \dot m_k=-2\langle u_x\rangle(x_k)=2\sum_{j=1}^n\sgn(x_k-x_j)m_je^{-|x_j-x_k|}, \label{DP_eq:peakon}
\end{align}
where $\langle \cdot \rangle(x_j)$ denotes the arithemetic average of left and right limits at the point $x_j$. Recall that the DP equation admits the Lax pair
\begin{subequations}
 \begin{align}
 &(\partial_x-\partial_x^3)\psi=zm\psi,\label{dp_lax_x}\\
 &\psi_t=[z^{-1}(1-\partial_x^2)+u_x-u\partial_x]\psi.\label{dp_lax_t}
 \end{align}
\end{subequations}
Due to Lax integrability in the peakon sector, Lundmark and Szmigielski \cite{lundmark2003multi,lundmark2005degasperis} employed inverse spectral method to give an explicit construction of DP multipeakons. Now let's sketch their idea below.

Firstly, for the initial data $\{x_k(0),m_k(0)\}_{k=1}^n$ satisfying 
$$x_1(0)<x_2(0)<\cdots<x_n(0),\qquad m_k(0)>0,$$
they considered a discrete cubic string problem in a finite interval $[-1,1]$ related to the linear spectral problem \eqref{dp_lax_x} by a Liouville transformation \cite[Th. 3.1]{lundmark2005degasperis}, 
\begin{align*}
&− \phi_{yyy}(y) = zg(y)\phi(y), \qquad y\in (-1,1)\\ 
& \phi(−1) = \phi_y(−1) = 0,\quad  \phi(1) = 0,
\end{align*}
where 
$$g(y)=\sum_{k=1}^ng_k\delta(y-y_k),$$
with
\begin{equation}
g_k=8m_k\cosh^4  \frac{x_k}{2}>0, \ \ \ y_k=\tanh \frac{x_k}{2}.
\end{equation}
It was shown,  in \cite[Th. 3.3]{lundmark2005degasperis}, that  the discrete cubic string has $n$ distinct positive eigenvalues $\{\zeta_k\}_{k=1}^n$. By introducing a pair of Weyl functions, the extended spectral data $\{\zeta_k,a_k\}_{k=1}^n$ were found (see  \cite[Th. 3.5]{lundmark2005degasperis} and note that we omit the data $c_j$ originally appearing in their paper since $c_j$ can be explicitly expressed in terms of  $\zeta_j$ and $a_j$, which is described in \cite[Coro. 3.6]{lundmark2005degasperis}). 

Then, it turns out that there is a bijection between discrete cubic strings with $\{g_k,y_k\}_{k=1}^n$ 
restricted by
$$g_k>0,\qquad -1=y_0<y_1<y_2<\cdots<y_n<y_{n+1}=1$$
and the spectra data $\{\zeta_k,a_k\}_{k=1}^n$ 
satisfying
$$a_k>0,\qquad 0<\zeta_1<\zeta_2<\cdots<\zeta_n.$$
In fact, the inverse spectral mapping \cite[Th. 4.16]{lundmark2005degasperis} could be formulated explicitly by
\begin{equation}\label{exp:gy_uvw}
g_{k'}=\frac{(U_k+V_{k-1})^4}{2W_{k-1}W_k},\qquad y_{k'}=\frac{U_k-V_{k-1}}{U_k+V_{k-1}},
\end{equation}
with the index abbreviations $k'=n+1-k,\ k=1,2,\dots, n$. Here $U_k,V_k,W_k$ are defined as
\begin{align}
U_k=\sum_{I\in{\left(\substack{[1,n]\\k}\right)}}\frac{\Delta_I^2}{\Gamma_I}a_I, \quad  V_k=\sum_{I\in{\left(\substack{[1,n]\\k}\right)}}\frac{\Delta_I^2}{\Gamma_I}\zeta_Ia_I,\quad W_k=V_kU_k-V_{k-1}U_{k+1} \label{exp:uvw}
\end{align}
with the notations 
 \begin{equation*}
    a_I=\prod_{j\in I}a_j,
    \qquad
    \zeta_I^l=\prod_{j\in I}\zeta_j^l,
    \qquad
    \Delta_I=\prod_{i,j\in I, \, i<j}(\zeta_j-\zeta_i),\qquad
    \Gamma_I=\prod_{i,j\in I, \, i<j}(\zeta_j+\zeta_i).
  \end{equation*}
and
  \begin{equation*}
    \binom{[1,K]}{k}=\{J=\{j_1,\dots,j_k\} : 1\leq j_1<\dots<j_k\leq K\}
  \end{equation*}
denoting the set of $k$-element subsets
  $J=\{j_1<\dots<j_k\}$ of the integer interval
  $[1,K]=\{1,\dots,K\}$. It is noted that we have the convention that $U_0=V_0=W_0=1$ and $U_k=V_k=W_k=0$ for $k<0$ and $k>n$. And it is obvious that $U_k>0,\ V_k>0$ for $1\leq k\leq n$. The positivity of  $W_k$ for $1\leq k \leq n$ is claimed in \cite[Lemma 2.20]{lundmark2005degasperis}.

Moreover, the $t$ part of the Lax pair  \eqref{dp_lax_t} implies that the spectral data evolve linearly \cite[Th. 2.15]{lundmark2005degasperis},  that is,
$$\dot\zeta_k=0,\qquad \dot a_k=\frac{a_k}{\zeta_k},$$
so that $\zeta_k$ are positive constants and $a_k(t)=a_k(0)e^{\frac{t}{\zeta_k}}>0$.

To sum up, it is not hard to conclude that 

\begin{theorem}[Lundmark \& Szmigielski {\cite[Th. 2.23]{lundmark2005degasperis}}]\label{th:DP_solution}
The DP equation admits the (at least local) n-peakon solution  of the form
\begin{equation*}
u=\sum_{k=1}^n m_k(t)e^{-|x-x_k(t)|},
\end{equation*}
where
\begin{equation}\label{sol:DPxt_form}
x_{k'}=\log\frac{U_k}{V_{k-1}},\qquad m_{k'}=\frac{(U_k)^2(V_{k-1})^2}{W_kW_{k-1}},
\end{equation}
for the index $k'=n+1-k,\ k=1,2,\dots, n$. Here
\begin{align*}
U_k=\sum_{I\in{\left(\substack{[1,n]\\k}\right)}}\frac{\Delta_I^2}{\Gamma_I}a_I, \quad  V_k=\sum_{I\in{\left(\substack{[1,n]\\k}\right)}}\frac{\Delta_I^2}{\Gamma_I}\zeta_Ia_I,\quad W_k=V_kU_k-V_{k-1}U_{k+1}
\end{align*}
with the constants $\zeta_j$ and $a_j(t)$ satisfying
\begin{equation}\label{DP_evl:b}
0<\zeta_1<\zeta_2<\cdots\zeta_n,\qquad 
\dot a_j(t)=\frac{a_j(t)}{\zeta_j}>0.
\end{equation}
\end{theorem}
\begin{remark}
Let 
$$
\mathcal{P}=\{(x_j,m_j)\ |\ x_1<x_2<\cdots<x_n,\ \ \  m_j>0,\ \ \  j=1,2,\dots,n\}.
$$
 It was shown that, if the initial data are in the space $\mathcal{P}$, then $x_j(t),m_j(t)$ will exist and remain in the space  $\mathcal{P}$ for all the time $t\in \bf R$ under the peakon flow \eqref{DP_eq:peakon}. In other words, Theorem \ref{th:DP_solution} describes a global n-peakon solution.
\end{remark}
\begin{remark}
Since all the momenta $m_j$ are positive, Theorem \ref{th:DP_solution} gives the pure n-peakon solution. It is not hard to formulate the result for the negative momenta by using a trivial transformation, which corresponds to the pure antipeakon case.
\end{remark}

\subsection{Determinant formulae}  The explicit construction of the inverse mapping \eqref{exp:gy_uvw} in \cite{lundmark2005degasperis} is not trivial and somewhat intricate. Its solution involves the idea of solving a Hermite-Pad\'e type approximation problem, which has a unique solution and could be expressed in terms of certain bimoment determinants with respect to the Cauchy kernel. And then, evaluations of the bimoment determinants lead to the explicit formulae \eqref{exp:gy_uvw}. In this subsection, the bimoment determinants in terms of Cauchy kernel are introduced and some formulae are derived, which are helpful for us to deduce the DP peakon lattice and the C-Toda lattice and their connection. It is noted that some of the formulae have appeared in \cite{lundmark2005degasperis}, while most of them are new.
Also note that the time dependence is not involved in this subsection.

Introduce the bimoments involving the Cauchy kernel
\begin{equation*}
I_{i,j}=I_{j,i}=\iint_{\mathbb{R}_+^2} \frac{x^iy^j}{x+y}d\mu(x) d\mu(y)
\end{equation*}
and the single moments 
\begin{equation}
\alpha_i=\int_{\mathbb{R_+}} x^id\mu(x),
\end{equation}
where $\mu$ is a discrete measure on $\mathbb{R}_+$
\[
\mu=\sum_{p=1}^na_p\delta_{\zeta_p}, \qquad  \text{with} \qquad 0<\zeta_1<\zeta_2<\cdots<\zeta_n,\quad a_p>0.
\]
Obviously, the bimoments $I_{i,j}$ and the single moments $\alpha_i$ admit the explicit expressions as
\[
I_{i,j}=I_{j,i}=\sum_{p=1}^n\sum_{q=1}^n\frac{\zeta_p^i\zeta_q^j}{\zeta_p+\zeta_q}a_pa_q,\qquad \alpha_i=\sum_{p=1}^n\zeta_p^ia_p.
\]
Let us consider the following determinants ($F_k^{(i,j)},G_k^{(i,j)},E_k^{(i,j)}$) with the bimoments $I_{i,j}$ and the single moments $\alpha_i$ as elements:
\begin{define}
For $k\geq1$, let $F_k^{(i,j)}$ denote the determinant of the $k\times k$ bimoment matrix which starts with $I_{i,j}$ at the upper left corner:
\begin{align}
F_k^{(i,j)}=\left|
\begin{array}{cccc}
I_{i,j}&I_{i,j+1}&\cdots & I_{i,j+k-1}\\
I_{i+1,j}&I_{i+1,j+1}&\cdots&I_{i+1,j+k-1}\\
\vdots&\vdots&\ddots&\vdots\\
I_{i+k-1,j}&I_{i+k-1,j+1}&\cdots&I_{i+k-1,j+k-1}
\end{array}
\right|=F_k^{(j,i)}.
\end{align}
Let $F_0^{(i,j)}=1$ and $F_k^{(i,j)}=0$ for $k<0$.

For $k\geq2$, let $G_k^{(i,j)}$ denote the $k\times k$ determinant
\begin{align}
G_k^{(i,j)}=\left|
\begin{array}{ccccc}
\alpha_{i-1}&I_{i,j}&I_{i,j+1}&\cdots & I_{i,j+k-2}\\
\alpha_i&I_{i+1,j}&I_{i+1,j+1}&\cdots&I_{i+1,j+k-2}\\
\vdots&\vdots&\vdots&\ddots&\vdots\\
\alpha_{i+k-2}&I_{i+k-1,j}&I_{i+k-1,j+1}&\cdots&I_{i+k-1,j+k-2}
\end{array}
\right|.
\end{align}
Let $G_1^{(i,j)}=\alpha_{i-1}$ and $G_k^{(i,j)}=0$ for $k<1$.

For $k\geq2$, let $E_k^{(i,j)}$ denote the $k\times k$ determinant
\begin{align}
E_k^{(i,j)}=\left|
\begin{array}{ccccc}
0&\alpha_{j-1}&\alpha_j&\cdots&\alpha_{j+k-3}\\
\alpha_{i-1}&I_{i,j}&I_{i,j+1}&\cdots & I_{i,j+k-2}\\
\alpha_i&I_{i+1,j}&I_{i+1,j+1}&\cdots&I_{i+1,j+k-2}\\
\vdots&\vdots&\vdots&\ddots&\vdots\\
\alpha_{i+k-3}&I_{i+k-2,j}&I_{i+k-2,j+1}&\cdots&I_{i+k-2,j+k-2}
\end{array}
\right|=E_k^{(j,i)}.
\end{align}
Let $E_k^{(i,j)}=0$ for $k<2$.
\end{define}
\begin{remark}\label{rem:multiple}
The determinants $F_k^{(i,j)}$ own multiple integral representations 
$$
F_k^{(i,j)}=\displaystyle\iint_{\substack{0<x_1<\cdots<x_n\\0<y_1<\cdots<y_n}}\frac{\prod_{1\leq p<q\leq n}(x_p-x_q)^2(y_p-y_q)^2}{\prod_{p,q=1}^n{(x_p+y_q)}}\prod_{p=1}^n\prod_{q=1}^n(x_p)^i(y_q)^jd\mu(x_p)d\mu(y_q),
$$ which are closely related to
the partition function of the Cauchy two-matrix model \cite{bertola2009cauchy,bertola2014cauchy}. 
\end{remark}

\begin{remark}
Note that, in contrast to the determinants in Section \ref{subsec:det}, there is offset of the index of the
moments ($\alpha_{i-1}$ etc.) here.
\end{remark}

As is shown in \cite{hone2009explicit,lundmark2005degasperis}, for some specific $(i,j)$, $F_k^{(i,j)},G_k^{(i,j)}$ could be explicitly evaluated in terms of $U_k, V_k, W_k$ in \eqref{exp:uvw}. These relations are useful for us so that the expressions in Theorem \ref{th:DP_solution} can be rewritten in terms of $F_k^{(i,j)},G_k^{(i,j)}$. For our convenience, the objects we will deal with are these bimoment determinants instead of  $U_k, V_k, W_k$. 
\begin{lemma} \label{lem:FGUVW}
For $1\leq k\leq n$, there hold 
\begin{align*}
&F_k^{(1,0)}=F_k^{(0,1)}=\frac{(U_k)^2}{2^k}>0,&& F_k^{(1,1)}=\frac{W_k}{2^k}>0,&& F_k^{(2,1)}=F_k^{(1,2)}=\frac{(V_k)^2}{2^k}>0,\\
&G_k^{(1,0)}=\frac{U_kU_{k-1}}{2^{k-1}}>0,&& G_k^{(1,1)}=\frac{U_kV_{k-1}}{2^{k-1}}>0,&& G_k^{(2,1)}=\frac{V_kV_{k-1}}{2^{k-1}}>0,\\
&G_k^{(2,0)}=\frac{U_{k-1}V_{k}-U_{k+1}V_{k-2}}{2^{k-1}}.
\end{align*}
\end{lemma}
\begin{proof}
All the formulae can be found in \cite{hone2009explicit,lundmark2005degasperis} except the last one for $G_k^{(2,0)}$. It immediately follows from the identity \eqref{id:bi7} with $i=1,j=0$ presented in Lemma \ref{lem:bi} and previously known formulae.  Furthermore, the positivity follows from the definition of $U_k, V_k$ or the integral representation of $F_k^{(i,j)}$ in Remark \ref{rem:multiple}.
\end{proof}
 From this lemma, the following corollary immediately follows.
\begin{coro} \label{coro:FG}
For any $k\in \mathbb{Z}$, there hold 
\begin{align}
(G_k^{(1,0)})^2=2F_k^{(1,0)}F_{k-1}^{(1,0)},\qquad (G_k^{(1,1)})^2=2F_k^{(1,0)}F_{k-1}^{(2,1)}. \label{rel:FG}
\end{align}
\end{coro}
\begin{remark} By using Lemma \ref{lem:FGUVW} and Corollary \ref{coro:FG}, the explicit formulae of the peakon solution \eqref{sol:DPxt_form}  of the DP equation could be rewritten in terms of $F_k^{(i,j)},G_k^{(i,j)}$, which will be presented in the next subsection (see \eqref{sol:DP_FG}) for convenience.
\end{remark}

An interesting observation is that the bimoment determinants $F_k^{(i,j)},G_k^{(i,j)},E_k^{(i,j)}$ of high order are null. This is not surprising because the measure $\mu$ has finite support. Particularly, $F_k^{(i,j)}=0$ for $k>n$ follows immediately from the integral representation in Remark \ref{rem:multiple}. Such phenomenon has appeared in the Hankel determinants with finite measures for CH peakon problems \cite{beals2000multipeakons,chang2014generalized,chang2016multipeakons}, and actually the corresponding result some specific $F_k^{(i,j)},G_k^{(i,j)}$ has been obtained in \cite{hone2009explicit,lundmark2005degasperis}.  In the following, we present a generic result with a proof based on matrix factorizations, which are motivated by those for Hankel determinants \cite{chang2014generalized,chang2016multipeakons} but much more complicated.
\begin{lemma} \label{lem:FGE0}
For $k>n$, there hold
\begin{equation}
F_k^{(i,j)}=G_k^{(i,j)}=E_{k+1}^{(i,j)}=0.
\end{equation}
\end{lemma}
\begin{proof}
By employing the expressions of $I_{i,j}$ and $\alpha_i$ in terms of $\zeta_p,a_p$, it is not hard to see
\begin{align*}
&\left(
\begin{array}{cccc}
I_{i,j}&I_{i,j+1}&\cdots & I_{i,j+k-1}\\
I_{i+1,j}&I_{i+1,j+1}&\cdots&I_{i+1,j+k-1}\\
\vdots&\vdots&\ddots&\vdots\\
I_{i+k-1,j}&I_{i+k-1,j+1}&\cdots&I_{i+k-1,j+k-1}
\end{array}
\right)
=\left(
\begin{array}{cccc}
\zeta_1^i&\zeta_2^i&\cdots&\zeta_n^i\\
\zeta_1^{i+1}&\zeta_2^{i+1}&\cdots&\zeta_n^{i+1}\\
\vdots&\vdots&\ddots&\vdots\\
\zeta_1^{i+k-1}&\zeta_2^{i+k-1}&\cdots&\zeta_n^{i+k-1}
\end{array}
\right)_{k\times n}\\
&\qquad\qquad\qquad\qquad
\cdot
\left(
\begin{array}{cccc}
\frac{a_1a_1}{\zeta_1+\zeta_1}&\frac{a_1a_2}{\zeta_1+\zeta_2}&\cdots&\frac{a_1a_n}{\zeta_1+\zeta_n}\\
\frac{a_2a_1}{\zeta_2+\zeta_1}&\frac{a_2a_2}{\zeta_1+\zeta_2}&\cdots&\frac{a_2a_n}{\zeta_2+\zeta_n}\\
\vdots&\vdots&\ddots&\vdots\\
\frac{a_na_1}{\zeta_n+\zeta_1}&\frac{a_na_2}{\zeta_n+\zeta_2}&\cdots&\frac{a_na_n}{\zeta_n+\zeta_n}
\end{array}
\right)_{n\times n}
\cdot
\left(
\begin{array}{cccc}
\zeta_1^j&\zeta_1^{j+1}&\cdots&\zeta_1^{j+k-1}\\
\zeta_2^j&\zeta_2^{j+1}&\cdots&\zeta_2^{j+k-1}\\
\vdots&\vdots&\ddots&\vdots\\
\zeta_n^j&\zeta_n^{j+1}&\cdots&\zeta_n^{j+k-1}
\end{array}
\right)_{n\times k},
\end{align*}
which imply that the rank of the matrix on the left-hand side is not more than $n$. Thus we have $F_k^{(i,j)}=0$ for $k>n$.

On the other hand, since there hold the following matrix factorizations:
\begin{align*}
&\left(
\begin{array}{ccccc}
\alpha_{i-1}&I_{i,j}&I_{i,j+1}&\cdots & I_{i,j+k-2}\\
\alpha_i&I_{i+1,j}&I_{i+1,j+1}&\cdots&I_{i+1,j+k-2}\\
\vdots&\vdots&\vdots&\ddots&\vdots\\
\alpha_{i+k-2}&I_{i+k-1,j}&I_{i+k-1,j+1}&\cdots&I_{i+k-1,j+k-2}
\end{array}
\right)=\left(
\begin{array}{cccc}
0&\zeta_1^i&\cdots&\zeta_n^i\\
0&\zeta_1^{i+1}&\cdots&\zeta_n^{i+1}\\
\vdots&\vdots&\ddots&\vdots\\
0&\zeta_1^{i+k-1}&\cdots&\zeta_n^{i+k-1}
\end{array}
\right)_{k\times (n+1)} \\
&\qquad\qquad\qquad\qquad
\cdot
\left(
\begin{array}{cccc}
0&0&\cdots&0\\
\frac{a_1}{\zeta_1}&\frac{a_1a_1}{\zeta_1+\zeta_1}&\cdots&\frac{a_1a_n}{\zeta_1+\zeta_n}\\
\vdots&\vdots&\ddots&\vdots\\
\frac{a_n}{\zeta_n}&\frac{a_na_1}{\zeta_n+\zeta_1}&\cdots&\frac{a_na_n}{\zeta_n+\zeta_n}
\end{array}
\right)_{(n+1)\times (n+1)}
\cdot
\left(
\begin{array}{cccc}
1&0&\cdots&0\\
0&\zeta_1^j&\cdots&\zeta_1^{j+k-2}\\
\vdots&\vdots&\ddots&\vdots\\
0&\zeta_n^j&\cdots&\zeta_n^{j+k-2}
\end{array}
\right)_{(n+1)\times k},
\end{align*}
\begin{align*}
&\left(
\begin{array}{ccccc}
0&\alpha_{j-1}&\alpha_j&\cdots&\alpha_{j+k-3}\\
\alpha_{i-1}&I_{i,j}&I_{i,j+1}&\cdots & I_{i,j+k-2}\\
\alpha_i&I_{i+1,j}&I_{i+1,j+1}&\cdots&I_{i+1,j+k-2}\\
\vdots&\vdots&\vdots&\ddots&\vdots\\
\alpha_{i+k-3}&I_{i+k-2,j}&I_{i+k-2,j+1}&\cdots&I_{i+k-2,j+k-2}
\end{array}
\right)=\left(
\begin{array}{cccc}
1&0&\cdots&0\\
0&\zeta_1^i&\cdots&\zeta_n^i\\
\vdots&\vdots&\ddots&\vdots\\
0&\zeta_1^{i+k-2}&\cdots&\zeta_n^{i+k-2}
\end{array}
\right)_{k\times (n+1)} \\
&\qquad\qquad\qquad
\cdot
\left(
\begin{array}{cccc}
0&\frac{a_1}{\zeta_1}&\cdots&\frac{a_n}{\zeta_n}\\
\frac{a_1}{\zeta_1}&\frac{a_1a_1}{\zeta_1+\zeta_1}&\cdots&\frac{a_1a_n}{\zeta_1+\zeta_n}\\
\vdots&\vdots&\ddots&\vdots\\
\frac{a_n}{\zeta_n}&\frac{a_na_1}{\zeta_n+\zeta_1}&\cdots&\frac{a_na_n}{\zeta_n+\zeta_n}
\end{array}
\right)_{(n+1)\times (n+1)}
\cdot
\left(
\begin{array}{cccc}
1&0&\cdots&0\\
0&\zeta_1^j&\cdots&\zeta_1^{j+k-2}\\
\vdots&\vdots&\ddots&\vdots\\
0&\zeta_n^j&\cdots&\zeta_n^{j+k-2}
\end{array}
\right)_{(n+1)\times k},
\end{align*}
it follows that the rank of the matrix on the left-hand side is not more than $n+1$, which gives $G_k^{(i,j)}=E_k^{(i,j)}=0$ for $k>n+1$. Furthermore, for $G_{n+1}^{(i,j)}=0$, since all matrices are $(n + 1) \times (n + 1)$ in this case, we can just take determinants of each factor on the right-hand side and use that the middle determinant is obviously zero. Thus, the proof is completed.
 \end{proof}

There exist rich relations among $F_k^{(i,j)},G_k^{(i,j)},E_k^{(i,j)}$. In the following, we derive some bilinear identities, which play important roles in the subsequent content.  They are all the consequences of employing the well known Jacobi determinant  identity  \cite{aitken1959determinants}, which reads
\begin{align}
\mathcal{D} \mathcal{D}\left(\begin{array}{cc}
i_1 & i_2 \\
j_1 & j_2 \end{array}\right)=\mathcal{D}\left(\begin{array}{c}
i_1  \\
j_1 \end{array}\right)\mathcal{D}\left(\begin{array}{c}
i_2  \\
j_2 \end{array}\right)-\mathcal{D}\left(\begin{array}{c}
i_1  \\
j_2 \end{array}\right)\mathcal{D}\left(\begin{array}{c}
i_2  \\
j_1 \end{array}\right). \label{id:jacobi}
\end{align}
Here $\mathcal{D}$ is an indeterminate determinant. $\mathcal{D}\left(\begin{array}{cccc}
i_1&i_2 &\cdots& i_k\\
j_1&j_2 &\cdots& j_k
\end{array}\right)$ with $ i_1<i_2<\cdots<i_k,\ j_1<j_2<\cdots<j_k$ denotes the
determinant of the matrix obtained from $\mathcal{D}$ by removing the rows with indices
$i_1,i_2,\dots, i_k$ and the columns with indices $j_1,j_2,\dots, j_k$.
\begin{lemma} \label{lem:bi}
For any $i,j\in \mathbb{Z}$, $k\in  \mathbb{N}_+$, there hold 
\begin{align}
&F_{k+1}^{(i,j)}F_{k-1}^{(i+1,j+1)}=F_{k}^{(i,j)}F_{k}^{(i+1,j+1)}-F_{k}^{(i,j+1)}F_{k}^{(i+1,j)},\label{id:bi1}\\
&G_{k+1}^{(i,j)}F_{k-1}^{(i+1,j)}=G_{k}^{(i,j)}F_{k}^{(i+1,j)}-G_{k}^{(i+1,j)}F_{k}^{(i,j)},\label{id:bi2}\\
&G_{k}^{(i,j)}F_{k-1}^{(i,j+1)}=G_{k}^{(i,j+1)}F_{k-1}^{(i,j)}-G_{k-1}^{(i,j+1)}F_{k}^{(i,j)},\label{id:bi3}\\
&E_{k+1}^{(i,j-1)}F_{k-1}^{(i+1,j)}=F_{k}^{(i,j-1)}E_{k}^{(i+1,j)}-G_{k}^{(j-1,i+1)}G_{k}^{(i,j)},\label{id:bi5}\\
&G_{k}^{(i,j+1)}F_{k-1}^{(i+1,j)}=F_{k}^{(i,j)}G_{k-1}^{(i+1,j+1)}+F_{k-1}^{(i+1,j+1)}G_{k}^{(i,j)},\label{id:bi4}\\
&E_{k}^{(i,j+1)}F_{k-1}^{(i,j)}=G_{k}^{(j,i)}G_{k-1}^{(i,j+1)}+E_{k}^{(i,j)}F_{k-1}^{(i,j+1)}.\label{id:bi6}\\
&G_{k+1}^{(i,j)}G_{k-1}^{(i+1,j+1)}=G_{k}^{(i,j)}G_{k}^{(i+1,j+1)}-G_{k}^{(i,j+1)}G_{k}^{(i+1,j)}.\label{id:bi7}
\end{align}
\end{lemma}
\begin{proof}
We shall give the detailed proof for $k>1$ since the case for $k=1$ is trivial.
Take 
\begin{align*}
&\mathcal{D}_1=F_{k+1}^{(i,j)},\qquad \mathcal{D}_2=G_{k+1}^{(i,j)},\\
&\mathcal{D}_3=\left|
\begin{array}{ccccc}
0&1&0&\cdots&0\\
\alpha_{i-1}&I_{i,j}&I_{i,j+1}&\cdots & I_{i,j+k-1}\\
\alpha_i&I_{i+1,j}&I_{i+1,j+1}&\cdots&I_{i+1,j+k-1}\\
\vdots&\vdots&\vdots&\ddots&\vdots\\
\alpha_{i+k-2}&I_{i+k-1,j}&I_{i+k-1,j+1}&\cdots&I_{i+k-1,j+k-1}
\end{array}
\right|,\\
&i_1=j_1=1,\qquad i_2=j_2=k+1.
\end{align*}
Applying the Jacobi identity \eqref{id:jacobi} to $\mathcal{D}_1$, $\mathcal{D}_2$ and $\mathcal{D}_3$, respectively, one obtains \eqref{id:bi1}--\eqref{id:bi3}.

The Jacobi identity \eqref{id:jacobi} with the setting 
\begin{align*}
\mathcal{D}_4=E_{k+1}^{(i,j-1)},\qquad i_1=j_1=1,\qquad i_2=j_2=2
\end{align*}
can yield \eqref{id:bi5}.

If one considers 
\begin{align*}
&\mathcal{D}_5=\mathcal{D}_3,\quad i_1=j_1=1,\quad i_2=2, \quad j_2=k+1,\\
& \mathcal{D}_6= \left|
\begin{array}{ccccc}
0&1&0&\cdots&0\\
0&\alpha_{j-1}&\alpha_j&\cdots&\alpha_{j+k-2}\\
\alpha_{i-1}&I_{i,j}&I_{i,j+1}&\cdots & I_{i,j+k-1}\\
\alpha_i&I_{i+1,j}&I_{i+1,j+1}&\cdots&I_{i+1,j+k-1}\\
\vdots&\vdots&\vdots&\ddots&\vdots\\
\alpha_{i+k-3}&I_{i+k-2,j}&I_{i+k-2,j+1}&\cdots&I_{i+k-2,j+k-1}
\end{array}
\right|,
\end{align*}
then  \eqref{id:bi4}-\eqref{id:bi6} can be derived by use of the Jacobi identity \eqref{id:jacobi}. 

The last formula \eqref{id:bi7} is a consequence of applying the Jacobi identity \eqref{id:jacobi}
with 
$$ \mathcal{D}_7=\mathcal{D}_2,\quad i_1=1,\quad j_1=2, \quad i_2=j_2=k+1.$$

\end{proof}

By using these bilinear identities, one is led to the following nonlinear relations.
\begin{coro} \label{coro:idsum}
For any $i,j\in \mathbb{Z}$ and integer $k$ satisfying $1\leq k\leq n$,  there hold
\begin{align}
&\sum_{l=k}^n\frac{G_{l}^{(i,j)}F_{l}^{(i-1,j)}}{F_{l}^{(i,j)}F_{l-1}^{(i,j)}}=\frac{G_{k}^{(i-1,j)}}{F_{k-1}^{(i,j)}},\\
&\sum_{l=1}^{k}\frac{G_{l}^{(i,j)}F_{l-1}^{(i,j+1)}}{F_{l}^{(i,j)}F_{l-1}^{(i,j)}}=\frac{G_{k}^{(i,j+1)}}{F_{k}^{(i,j)}}.
\end{align}
\end{coro}
\begin{proof}
By using the bilinear identity \eqref{id:bi2} in Lemma \ref{lem:bi}, we have 
\[
\sum_{l=k}^n\frac{G_{l}^{(i,j)}F_{l}^{(i-1,j)}}{F_{l}^{(i,j)}F_{l-1}^{(i,j)}}=\sum_{l=k}^n\left(\frac{G_{l}^{(i-1,j)}}{F_{l-1}^{(i,j)}}-\frac{G_{l+1}^{(i-1,j)}}{F_{l}^{(i,j)}}\right)=\frac{G_{k}^{(i-1,j)}}{F_{k-1}^{(i,j)}},
\]
where the facts $F_{n+1}^{(i,j)}=0$ are used.

Based on the identity \eqref{id:bi3} in Lemma \ref{lem:bi} and the facts $G_{0}^{(i,j)}=0$, we also have
\[
\sum_{l=1}^{k}\frac{G_{l}^{(i,j)}F_{l-1}^{(i,j+1)}}{F_{l}^{(i,j)}F_{l-1}^{(i,j)}}=\sum_{l=1}^k\left(\frac{G_{l}^{(i,j+1)}}{F_{l}^{(i,j)}}-\frac{G_{l-1}^{(i,j+1)}}{F_{l-1}^{(i,j)}}\right)=\frac{G_{k}^{(i,j+1)}}{F_{k}^{(i,j)}}.
\]
Thus the proof is completed.
\end{proof}
\subsection{A direct confirmation for Th. \ref{th:DP_solution}} In this subsection, we will give a direct confirmation for Th. \ref{th:DP_solution} by using determinant technique based on the results in the above subsection. For convenience, let's first give the following lemma, which expresses the evolution of the determinant $F_k^{(i,j)}$ in terms of $E_{k}^{(i,j)}$.
\begin{lemma} \label{lem:der_F}
If
\begin{align*}
\dot \zeta_p=0,\qquad  \dot a_p=\frac{a_p}{\zeta_p},
\end{align*}
 then we have
\begin{align}
\dot F_k^{(i,j)}=-E_{k+1}^{(i,j)}.
\end{align}
\end{lemma}
\begin{proof}
It is not hard to see 
\[
\dot I_{i,j}=\sum_{p=1}^n\sum_{q=1}^n\frac{\zeta_p^i\zeta_q^j}{\zeta_p+\zeta_q}a_pa_q\left(\frac{1}{\zeta_p}+\frac{1}{\zeta_q}\right)=\sum_{p=1}^n\sum_{q=1}^n\zeta_p^{i-1}\zeta_q^{j-1}a_pa_q=\alpha_{i-1}\alpha_{j-1}.
\]
According to the rule of differentiating a determinant and expansion along one column, we have 
\begin{align*}
\dot F_k^{(i,j)}&=\sum_{p=0}^{k-1}\left|
\begin{array}{ccccccc}
I_{i,j}&\cdots&I_{i,j+p-1}&\alpha_{i-1}\alpha_{j+p-1}&I_{i,j+p+1}&\cdots & I_{i,j+k-1}\\
I_{i+1,j}&\cdots&I_{i+1,j+p-1}&\alpha_{i}\alpha_{j+p-1}&I_{i+1,j+p+1}&\cdots&I_{i+1,j+k-1}\\
\vdots&\ddots&\vdots&\vdots&\vdots&\ddots&\vdots\\
I_{i+k-1,j}&\cdots&I_{i+k-1,j+p-1}&\alpha_{i+k-2}\alpha_{j+p-1}&I_{i+k-1,j+p+1}&\cdots&I_{i+k-1,j+k-1}
\end{array}
\right|\\
&=\sum_{p=0}^{k-1}\alpha_{j+p-1}\sum_{q=0}^{k-1}(-1)^{p+q}\alpha_{i+q-1} F_k^{(i,j)}\binom{q+1}{p+1}\\
&=\sum_{p=0}^{k-1}\sum_{q=0}^{k-1}(-1)^{p+q}\alpha_{j+p-1}\alpha_{i+q-1}F_k^{(i,j)}\binom{q+1}{p+1},
\end{align*}
the last of which is nothing but the expansion formula of the determinant $-E_{k+1}^{(i,j)}$ along the first row and the first column. Thus we complete the proof.
\end{proof}

Now we are ready to confirm the validity of Th. \ref{th:DP_solution} by using determinant technique.
\begin{proof}[An alternative proof to Th. \ref{th:DP_solution}]
By using Lemma \ref{lem:FGUVW} and Corollary \ref{coro:FG}, it is not hard to see that the solution \eqref{sol:DPxt_form} of the DP peakon system admits equivalent expressions as:
\begin{align}
x_{k'}=\log\frac{2F_k^{(1,0)}}{G_k^{(1,1)}}=\log\frac{G_k^{(1,1)}}{F_{k-1}^{(2,1)}}=\frac{1}{2}\log\frac{2F_k^{(1,0)}}{F_{k-1}^{(2,1)}},\qquad m_{k'}=\frac{(G_{k}^{(1,1)})^2}{2F_{k}^{(1,1)}F_{k-1}^{(1,1)}}=\frac{F_{k}^{(1,0)}F_{k-1}^{(2,1)}}{F_{k}^{(1,1)}F_{k-1}^{(1,1)}}, \label{sol:DP_FG}
\end{align}
which will be used in our setup. What we need to confirm is $\{x_k,m_k\}_{k=1}^n$ given by \eqref{sol:DP_FG} with time dependence \eqref{DP_evl:b} satisfy the ODE system \eqref{DP_eq:peakon}.

Recall that the DP peakon ODE system \eqref{DP_eq:peakon} reads 
\begin{subequations}
 \begin{align*}
  \dot x_{k'}&=u(x_{k'})=\sum_{j=1}^{k-1}m_{j'}e^{x_{k'}-x_{j'}}+\sum_{j=k}^nm_{j'}e^{x_{j'}-x_{k'}} ,\\
\dot m_{k'}&=-2m_{k'}\langle u_x(x_{k'})\rangle=2m_{k'}\left(\sum_{j=k+1}^nm_{j'}e^{x_{j'}-x_{k'}}-\sum_{j=1}^{k-1}m_{j'}e^{x_{k'}-x_{j'}}\right).
\end{align*}
\end{subequations}
By use of \eqref{sol:DP_FG}, we can see that it is sufficient to prove 
\begin{subequations}
 \begin{align*}
&\left(\frac{1}{2}\log\frac{2F_k^{(1,0)}}{F_{k-1}^{(2,1)}}\right)_t=\frac{G_k^{(1,1)}}{F_{k-1}^{(2,1)}}\sum_{j=1}^{k-1}\frac{G_{j}^{(1,1)}F_{j-1}^{(2,1)}}{2F_{j}^{(1,1)}F_{j-1}^{(1,1)}}+\frac{G_k^{(1,1)}}{2F_k^{(1,0)}}\sum_{j=k}^n\frac{F_{j}^{(1,0)}G_{j}^{(1,1)}}{F_{j}^{(1,1)}F_{j-1}^{(1,1)}},\\
&\left(\frac{1}{2}\log\frac{F_{k}^{(1,0)}F_{k-1}^{(2,1)}}{F_{k}^{(1,1)}F_{k-1}^{(1,1)}}\right)_t=\frac{G_k^{(1,1)}}{2F_k^{(1,0)}}\sum_{j=k+1}^n\frac{F_{j}^{(1,0)}G_{j}^{(1,1)}}{F_{j}^{(1,1)}F_{j-1}^{(1,1)}}-\frac{G_k^{(1,1)}}{F_{k-1}^{(2,1)}}\sum_{j=1}^{k-1}\frac{G_{j}^{(1,1)}F_{j-1}^{(2,1)}}{2F_{j}^{(1,1)}F_{j-1}^{(1,1)}},
\end{align*}
which yields 
\begin{align*}
&-\frac{E_{k+1}^{(1,0)}}{2F_k^{(1,0)}}+\frac{E_{k}^{(2,1)}}{2F_{k-1}^{(2,1)}}=\frac{G_k^{(1,1)}}{F_{k-1}^{(2,1)}}\cdot\frac{G_{k-1}^{(1,2)}}{2F_{k-1}^{(1,1)}}+\frac{G_k^{(1,1)}}{2F_k^{(1,0)}}\cdot\frac{G_k^{(0,1)}}{F_{k-1}^{(1,1)}},\\
&-\frac{E_{k+1}^{(1,0)}}{2F_k^{(1,0)}}-\frac{E_{k}^{(2,1)}}{2F_{k-1}^{(2,1)}}+\frac{E_{k+1}^{(1,1)}}{2F_k^{(1,1)}}+\frac{E_{k}^{(1,1)}}{2F_{k-1}^{(1,1)}}=\frac{G_k^{(1,1)}}{2F_k^{(1,0)}}\cdot\frac{G_{k+1}^{(0,1)}}{F_{k}^{(1,1)}}-\frac{G_k^{(1,1)}}{F_{k-1}^{(2,1)}}\cdot\frac{G_{k-1}^{(1,2)}}{2F_{k-1}^{(1,1)}},
\end{align*}
by employing Lemma \ref{lem:der_F} and Corollary \ref{coro:idsum}. 
\end{subequations}
A further simplification leads to 
 \begin{align*}
 &F_{k-1}^{(1,1)}(-E_{k+1}^{(1,0)}F_{k-1}^{(2,1)}+E_{k}^{(2,1)}F_k^{(1,0)})=G_k^{(1,1)}(G_{k-1}^{(1,2)}F_k^{(1,0)}+F_{k-1}^{(2,1)}G_k^{(0,1)}),\\
 &\frac{E_{k+1}^{(1,1)}F_k^{(1,0)}-E_{k+1}^{(1,0)}F_k^{(1,1)}}{F_k^{(1,0)}F_k^{(1,1)}}-\frac{E_{k}^{(2,1)}F_{k-1}^{(1,1)}-E_{k}^{(1,1)}F_{k-1}^{(2,1)}}{F_{k-1}^{(2,1)}F_{k-1}^{(1,1)}}=\frac{G_k^{(1,1)}G_{k+1}^{(0,1)}}{F_k^{(1,0)}F_{k}^{(1,1)}}-\frac{G_k^{(1,1)}G_{k-1}^{(1,2)}}{F_{k-1}^{(2,1)}F_{k-1}^{(1,1)}}.
 \end{align*}
It is not difficult to see that they are valid by observing the identities \eqref{id:bi5}, \eqref{id:bi4} and \eqref{id:bi6} in Lemma \ref{lem:bi}. Therefore, we complete the proof.
\end{proof}

Clearly, our proof implies that the DP peakon lattice can be seen as an isospectral flow on a manifold cut out by determinant identities. Recall that the CH peakon lattice is a flow on a manifold cut out by determinant identities \cite{chang2014generalized}, while the Novikov peakon lattice corresponds to a manifold cut out by Pfaffian identities \cite{chang2017application}. The difference arises from the different structures of the corresponding tau functions belonging to different Lie algebra of A,B,C types.
\section{DP peakon and C-Toda lattices}\label{sec:ctoda}
Recall that the CH peakon lattice is intimately linked to the finite A-Toda lattice\footnote{The A-Toda lattice denotes the ordinary Toda lattice.  Due to the AKP type, this makes sense and forms a contrast to B,C-Toda lattices.} \cite{beals2001peakons}, while the Novikov peakon lattice is related to the finite B-Toda lattice \cite{chang2017application}.  This section is devoted to deriving the similar relevance for the DP peakon lattice and the associated C-Toda lattice.

As is indicated in the above section, the DP peakon ODE system \eqref{DP_eq:peakon} can be linearised by use of inverse spectral method into 
\[
\dot \zeta_p(t)=0, \qquad \dot a_p(t)=\frac{a_p(t)}{\zeta_p}.
\]
Consequently, the moments
\[
\alpha_i(t)=\sum_{p=1}^n\zeta_p^ia_p(t),\quad I_{i,j}(t)=\sum_{p=1}^n\sum_{q=1}^n\frac{\zeta_p^i\zeta_q^j}{\zeta_p+\zeta_q}a_p(t)a_q(t),
\]
evolve according to
\[
\dot \alpha_i=\alpha_{i-1},\quad \dot I_{i,j}=\alpha_{i-1}\alpha_{j-1}.
\]
In some sense, this can be viewed as a negative flow.

Now we plan to start from a linearised flow in the positive direction 
\begin{align*}
&\dot \xi_p(t)=0, && 0<\xi_1<\cdots<\xi_n,\\
&\dot c_p(t)=\xi_pc_p(t), &&c_p(0)>0,
\end{align*}
to seek a nonlinear ODE system. 

\subsection{Determinant formulae} \label{subsec:det}
Introduce the moments
\begin{equation}
\beta_i(t)=\sum_{p=1}^n\xi_p^ic_p(t),\qquad J_{i,j}(t)=J_{j,i}(t)=\sum_{p=1}^n\sum_{q=1}^n\frac{\xi_p^i\xi_q^j}{\xi_p+\xi_q}c_p(t)c_q(t),
\end{equation}
and consider the bimoment determinants $\tau_k^{(i,j)}$, $\sigma_k^{(i,j)}$, $\rho_k^{(i,j)}$, $\omega_k^{(i,j)}$, $\nu_k^{(i,j)}$ of size $k\times k$ respectively defined by
\begin{align}\label{def:tau}
\tau_k^{(i,j)}=\left|
\begin{array}{cccc}
J_{i,j}&J_{i,j+1}&\cdots & J_{i,j+k-1}\\
J_{i+1,j}&J_{i+1,j+1}&\cdots&J_{i+1,j+k-1}\\
\vdots&\vdots&\ddots&\vdots\\
J_{i+k-1,j}&J_{i+k-1,j+1}&\cdots&J_{i+k-1,j+k-1}
\end{array}
\right|=\tau_k^{(j,i)}
\end{align}
with the convention $\tau_0^{(i,j)}=1$ and $\tau_k^{(i,j)}=0$ for $k<0$, and
\begin{align}
\sigma_k^{(i,j)}=\left|
\begin{array}{ccccc}
\beta_{i}&J_{i,j}&J_{i,j+1}&\cdots & J_{i,j+k-2}\\
\beta_{i+1}&J_{i+1,j}&J_{i+1,j+1}&\cdots&J_{i+1,j+k-2}\\
\vdots&\vdots&\vdots&\ddots&\vdots\\
\beta_{i+k-1}&J_{i+k-1,j}&J_{i+k-1,j+1}&\cdots&J_{i+k-1,j+k-2}
\end{array}
\right|
\end{align}
with the convention $\sigma_1^{(i,j)}=\beta_{i}$ and $\sigma_k^{(i,j)}=0$ for $k<1$, and
\begin{align}
\rho_k^{(i,j)}=\left|
\begin{array}{ccccc}
0&\beta_{j}&\beta_{j+1}&\cdots&\beta_{j+k-2}\\
\beta_{i}&J_{i,j}&J_{i,j+1}&\cdots & J_{i,j+k-2}\\
\beta_{i+1}&J_{i+1,j}&J_{i+1,j+1}&\cdots&J_{i+1,j+k-2}\\
\vdots&\vdots&\vdots&\ddots&\vdots\\
\beta_{i+k-2}&J_{i+k-2,j}&J_{i+k-2,j+1}&\cdots&J_{i+k-2,j+k-2}
\end{array}
\right|=\rho_k^{(j,i)}
\end{align}
with the convention $\rho_k^{(i,j)}=0$ for $k<2$, and
\begin{align}
\omega_k^{(i,j)}=\left|
\begin{array}{ccccc}
J_{i,j}&J_{i,j+1}&\cdots & J_{i,j+k-2}&J_{i,j+k}\\
J_{i+1,j}&J_{i+1,j+1}&\cdots&J_{i+1,j+k-2}&J_{i+1,j+k}\\
\vdots&\vdots&\ddots&\vdots&\vdots\\
J_{i+k-1,j}&J_{i+k-1,j+1}&\cdots&J_{i+k-1,j+k-2}&J_{i+k-1,j+k}
\end{array}
\right|
\end{align}
with the convention $\omega_1^{(i,j)}=J_{i,j+1}$ and $\omega_k^{(i,j)}=0$ for $k<1$, and
\begin{align}
\nu_k^{(i,j)}=\left|
\begin{array}{cccccc}
0&\beta_{j}&\beta_{j+1}&\cdots&\beta_{j+k-3}&\beta_{j+k-1}\\
\beta_{i}&J_{i,j}&J_{i,j+1}&\cdots & J_{i,j+k-3}& J_{i,j+k-1}\\
\beta_{i+1}&J_{i+1,j}&J_{i+1,j+1}&\cdots&J_{i+1,j+k-3}& J_{i+1,j+k-1}\\
\vdots&\vdots&\vdots&\ddots&\vdots&\vdots\\
\beta_{i+k-2}&J_{i+k-2,j}&J_{i+k-2,j+1}&\cdots&J_{i+k-2,j+k-3}& J_{i+k-2,j+k-1}
\end{array}
\right|
\end{align}
with the convention  $\nu_k^{(i,j)}=0$ for $k<2$. Note that, for the determinant formulae of $\omega_k^{(i,j)}$ and $\nu_{k}^{(i,j)}$, the index jumps an extra step between the last two columns, which differ from those for $\tau_k^{(i,j)}$ and $\rho_k^{(i,j)}$ respectively.

By using the Jacobi identity \eqref{id:jacobi} , it is not hard to see the following bilinear identities hold.
\begin{lemma}\label{lem:bi_id_tau}
For any $i,j\in \mathbb{Z}$, $k\in  \mathbb{N}_+$, there hold
\begin{align}
&\rho_{k+1}^{(i,j)}\tau_{k-1}^{(i,j)}=\tau_{k}^{(i,j)}\rho_{k}^{(i,j)}-\sigma_{k}^{(i,j)}\sigma_{k}^{(j,i)},\label{bi_id_tau1}\\
&\sigma_{k}^{(j,i)}\sigma_{k-1}^{(i,j)}=\nu_{k}^{(i,j)}\tau_{k-1}^{(i,j)}-\rho_{k}^{(i,j)}\omega_{k-1}^{(i,j)}.\label{bi_id_tau2}
\end{align}
\end{lemma}
\begin{proof}
Taking 
$$
\mathcal{D}_1=\rho_{k+1}^{(i,j)}, \qquad i_1=j_1=1, \quad i_2=j_2=k+1,
$$
we can obtain \eqref{bi_id_tau1}  by use of the Jacobi identity  \eqref{id:jacobi} .

If we consider
\begin{align*}
&\mathcal{D}_2=
\left|
\begin{array}{ccccc}
1&0&0&\cdots&0\\
0&\beta_{j}&\beta_{j+1}&\cdots&\beta_{j+k-1}\\
\beta_{i}&J_{i,j}&J_{i,j+1}&\cdots & J_{i,j+k-1}\\
\beta_{i+1}&J_{i+1,j}&J_{i+1,j+1}&\cdots&J_{i+1,j+k-1}\\
\vdots&\vdots&\vdots&\ddots&\vdots\\
\beta_{i+k-2}&J_{i+k-2,j}&J_{i+k-2,j+1}&\cdots&J_{i+k-2,j+k-1}
\end{array}
\right|,\\
&\quad  i_1=1, \quad j_1=k, \quad i_2=2,\quad j_2=k+1,
\end{align*}
then \eqref{bi_id_tau2} will be derived.
\end{proof}
\subsection{C-Toda lattice and its Lax pair}
If 
\begin{align}
\dot \xi_p=0, \qquad \dot c_p=\xi_pc_p, \label{evo_xib}
\end{align}
then $ \beta_i,J_{i,j}$ evolve according to
\begin{equation}
\dot \beta_i=\beta_{i+1},\quad \dot J_{i,j}=\beta_{i}\beta_{j}=J_{i+1,j}+J_{i,j+1}.
\end{equation}
Moreover, the evolution of the determinants $\tau_k^{(i,j)}$ could be expressed in terms of certain closed forms.
\begin{lemma} Under the condition \eqref{evo_xib}, the determinants $\tau_k^{(i,j)}$ admit the evolution:
\begin{align}
&\dot \tau_k^{(i,j)}=-\rho_{k+1}^{(i,j)}. \label{evo:tau1}
\end{align}
Besides, there also hold
\begin{align}
&\dot \tau_k^{(i,j)}=\omega_{k}^{(i,j)}+\omega_{k}^{(j,i)},\label{evo:tau2}\\
&\ddot \tau_k^{(i,j)}=-\nu_{k+1}^{(i,j)}-\nu_{k+1}^{(j,i)}.\label{evo:tau12}
\end{align}

\end{lemma}\label{lem:der_tau}
\begin{proof}
First we make use of the evolution relation
$$\dot J_{i,j}=\beta_{i}\beta_{j}.$$
According to the rule of differentiating a determinant and expansion along one column, we have 
\begin{align*}
\dot \tau_k^{(i,j)}&=\sum_{p=0}^{k-1}\left|
\begin{array}{ccccccc}
J_{i,j}&\cdots&J_{i,j+p-1}&\beta_{i}\beta_{j+p}&J_{i,j+p+1}&\cdots & J_{i,j+k-1}\\
J_{i+1,j}&\cdots&J_{i+1,j+p-1}&\beta_{i+1}\beta_{j+p}&J_{i+1,j+p+1}&\cdots&J_{i+1,j+k-1}\\
\vdots&\ddots&\vdots&\vdots&\vdots&\ddots&\vdots\\
J_{i+k-1,j}&\cdots&J_{i+k-1,j+p-1}&\beta_{i+k-1}\beta_{j+p}&J_{i+k-1,j+p+1}&\cdots&J_{i+k-1,j+k-1}
\end{array}
\right|\\
&=\sum_{p=0}^{k-1}\beta_{j+p}\sum_{q=0}^{k-1}(-1)^{p+q}\beta_{i+q} \tau_k^{(i,j)}\binom{q+1}{p+1}\\
&=\sum_{p=0}^{k-1}\sum_{q=0}^{k-1}(-1)^{p+q}\beta_{j+p}\beta_{i+q}\tau_k^{(i,j)}\binom{q+1}{p+1}.
\end{align*}
A closer check gives that the formula in the last row is no other than the expansion formula of the determinant $-\rho_{k+1}^{(i,j)}$ along the first row and the first column, which leads to \eqref{evo:tau1}.

In order to confirm \eqref{evo:tau2},  we shall employ the evolution 
$$\dot J_{i,j}=J_{i+1,j}+J_{i,j+1}.$$
A direct calculation gives
\begin{align*}
\dot \tau_k^{(i,j)}&=\sum_{p=0}^{k-1}\left|
\begin{array}{ccccccc}
J_{i,j}&\cdots&J_{i,j+p-1}&J_{i+1,j+p}&J_{i,j+p+1}&\cdots & J_{i,j+k-1}\\
J_{i+1,j}&\cdots&J_{i+1,j+p-1}&J_{i+2,j+p}&J_{i+1,j+p+1}&\cdots&J_{i+1,j+k-1}\\
\vdots&\ddots&\vdots&\vdots&\vdots&\ddots&\vdots\\
J_{i+k-1,j}&\cdots&J_{i+k-1,j+p-1}&J_{i+k,j+p}&J_{i+k-1,j+p+1}&\cdots&J_{i+k-1,j+k-1}
\end{array}
\right|\\
&+\sum_{p=0}^{k-1}\left|
\begin{array}{ccccccc}
J_{i,j}&\cdots&J_{i,j+p-1}&J_{i,j+p+1}&J_{i,j+p+1}&\cdots & J_{i,j+k-1}\\
J_{i+1,j}&\cdots&J_{i+1,j+p-1}&J_{i+1,j+p+1}&J_{i+1,j+p+1}&\cdots&J_{i+1,j+k-1}\\
\vdots&\ddots&\vdots&\vdots&\vdots&\ddots&\vdots\\
J_{i+k-1,j}&\cdots&J_{i+k-1,j+p-1}&J_{i+k-1,j+p+1}&J_{i+k-1,j+p+1}&\cdots&J_{i+k-1,j+k-1}
\end{array}
\right|\\
&=\mathcal{A}_1+\mathcal{A}_2,\\
\end{align*}
Regarding the first summation, we have 
\begin{align*}
\mathcal{A}_1=&\sum_{p=0}^{k-1}\sum_{q=0}^{k-1}(-1)^{p+q}J_{i+q+1,j+p} \tau_k^{(i,j)}\binom{q+1}{p+1}=\sum_{q=0}^{k-1}\sum_{p=0}^{k-1}(-1)^{p+q}J_{i+q+1,j+p} \tau_k^{(i,j)}\binom{q+1}{p+1}\\
=&\sum_{q=0}^{k-2}\sum_{p=0}^{k-1}(-1)^{p+q}J_{i+q+1,j+p} \tau_k^{(i,j)}\binom{q+1}{p+1}+\sum_{p=0}^{k-1}(-1)^{p+k-1}J_{i+k,j+p} \tau_k^{(i,j)}\binom{k}{p+1}.
\end{align*}
It follows from the Laplace expansion that, 
$$\sum_{p=0}^{k-1}(-1)^{p+q}J_{i+q+1,j+p} \tau_k^{(i,j)}\binom{q+1}{p+1}$$
 vanishes for fixed $0\leq q\leq k-2$ and subsequently
$$\sum_{p=0}^{k-1}(-1)^{p+k-1}J_{i+k,j+p} \tau_k^{(i,j)}\binom{k}{p+1}=\omega_k^{(j,i)}=\mathcal{A}_1.$$
As for the second summation, it is not hard to see that only the term for $p=k-1$ contributes and one is led to $$\mathcal{A}_2=\omega_k^{(i,j)}.$$ Thus we get 
$$\dot \tau_k^{(i,j)}=\mathcal{A}_1+\mathcal{A}_2=\omega_k^{(j,i)}+\omega_k^{(i,j)}.$$

Finally, let's turn to the proof of \eqref{evo:tau12}. It is sufficient to prove $$\dot \omega_k^{(i,j)}=-\nu_{k+1}^{(i,j)}.$$ This can be similarly derived by using 
$$\dot J_{i,j}=\beta_{i}\beta_{j}$$
as the process for  the proof to \eqref{evo:tau1}. We omit the detail here.
\end{proof}

Combining Lemma \ref{lem:bi_id_tau} and Lemma \ref{lem:der_tau}, we immediately obtain the following bilinear ODE system.
\begin{coro} \label{coro:tau_bi}
For any $j\in \mathbb{Z}$, $k\in  \mathbb{N}_+$, there hold
\begin{subequations}\label{eq:bi_id11}
\begin{align}
&\dot \tau_k^{(j,j)}\tau_{k-1}^{(j,j)}-\tau_k^{(j,j)} \dot \tau_{k-1}^{(j,j)}=\left(\sigma_k^{(j,j)} \right)^2,\\
&\ddot \tau_k^{(j,j)}  \tau_k^{(j,j)}-\left(\dot  \tau_k^{(j,j)}\right)^2=-2 \sigma_{k+1}^{(j,j)} \sigma_k^{(j,j)}.
\end{align}
\end{subequations}
\end{coro}
Eliminating the variables $\sigma_k^{(j,j)}$ in \eqref{eq:bi_id11} leads to a quartic lattice
\begin{align}\label{eq:ctoda_quartic}
\left(\ddot \tau_k^{(j,j)}  \tau_k^{(j,j)}-\left(\dot  \tau_k^{(j,j)}\right)^2\right)^2=4\left(\dot \tau_{k+1}^{(j,j)}\tau_{k}^{(j,j)}-\tau_{k+1}^{(j,j)} \dot \tau_{k}^{(j,j)}\right)\left(\dot \tau_k^{(j,j)}\tau_{k-1}^{(j,j)}-\tau_k^{(j,j)} \dot \tau_{k-1}^{(j,j)}\right).
\end{align}
{as we will show in Section \ref{sec:schief}, which is indeed related to the fully discrete CKP lattice proposed by Schief \cite{bobenko2015discrete,bobenko2017discrete,schief2003lattice}.}
If we introduce the variables 
$$u_k=\frac{\tau_{k+1}^{(j,j)}\tau_{k-1}^{(j,j)}}{({\tau_{k}^{(j,j)}})^2},\qquad b_k=\left(\log\frac{\tau_{k+1}^{(j,j)}}{\tau_{k}^{(j,j)}}\right)_t=\frac{(\sigma_{k+1}^{(j,j)})^2}{\tau_{k+1}^{(j,j)}\tau_{k}^{(j,j)}},$$
then it is not hard to obtain
 $\{u_k,b_k\}$ satisfy a nonlinear ODE system: 
\begin{subequations}
\begin{align}
&\dot u_{k}=u_{k}(b_{k}-b_{k-1}),  \qquad\qquad \qquad\qquad \ k=1,\dots,n-1,\\
&\dot b_k=2(\sqrt{ u_kb_kb_{k-1} }-\sqrt{ u_{k+1}b_{k+1}b_{k} }), \quad k=0,\dots,n-1,
\end{align}
\end{subequations}
with $u_0=u_n=0$. This is nothing but the nonlinear flow, which we are searching for from the positive linear flow \eqref{evo_xib}.
We shall call this ODE system ``finite Toda lattice of CKP type'' (finite C-Toda lattice) due to the similar form to the A-Toda lattice and the tau function structure of the CKP type (see Table \ref{comp_abc}), { and also because of the imitate relation with  the fully discrete CKP lattice in \cite{bobenko2015discrete,bobenko2017discrete,schief2003lattice}.}

Based on the above derivation, we can conclude that 
\begin{theorem}\label{th:ctoda}
The finite C-Toda lattice 
\begin{subequations}
\label{eq:ctoda}
\begin{align}
&\dot u_{k}=u_{k}(b_{k}-b_{k-1}),  \qquad\qquad \qquad\qquad \ k=1,\dots,n-1,\\
&\dot b_k=2(\sqrt{ u_kb_kb_{k-1} }-\sqrt{ u_{k+1}b_{k+1}b_{k} }), \quad k=0,\dots,n-1,
\end{align}
\end{subequations}
admit a solution
$$u_k=\frac{\tau_{k+1}^{(j,j)}\tau_{k-1}^{(j,j)}}{({\tau_{k}^{(j,j)}})^2},\qquad b_k=\frac{(\sigma_{k+1}^{(j,j)})^2}{\tau_{k+1}^{(j,j)}\tau_{k}^{(j,j)}},$$
where $\tau_{k}^{(j,j)}$ is defined as \eqref{def:tau} with 
$$J_{i,j}(t)=J_{j,i}(t)=\sum_{p=1}^n\sum_{q=1}^n\frac{\xi_p^i\xi_q^j}{\xi_p+\xi_q}c_p(t)c_q(t),\qquad 0<\xi_1<\cdots<\xi_n,\  c_p(t)=c_p(0)e^{\xi_pt}>0.$$
\end{theorem}
\begin{remark}
The readers might have observed that there is some freedom in the solutions. 
But, notice that, for any fixed $j$, $\tau_{k}^{(j,j)}$ defined on  a set of $\{\zeta_p,c_p\}_{p=1}^n$ may be regarded as a new determinant $\tau_{k}^{(1,1)}$ defined on a modified set of $\{\zeta_p,\zeta_{p}^{j-1}c_p\}_{p=1}^n$. This suggests that it is sufficient to consider the solution 
$$u_k=\frac{\tau_{k+1}^{(1,1)}\tau_{k-1}^{(1,1)}}{({\tau_{k}^{(1,1)}})^2},\qquad b_k=\frac{(\sigma_{k+1}^{(1,1)})^2}{\tau_{k+1}^{(1,1)}\tau_{k}^{(1,1)}}$$
to the finite C-Toda lattice \eqref{eq:ctoda} 
without loss of generality. In fact, for some initial value problem of \eqref{eq:ctoda}, we guess one can determine uniquely the set $\{\zeta_p,c_p\}_{p=1}^n$, which will be addressed in the future publication.
\end{remark}

The $\tau_{k}^{(1,1)}$ can be viewed as the tau function of the finite C-Toda lattice. Due to its structure, it is reasonable to inquire the finite C-Toda lattice can be regarded as the isospectral deformation of Cauchy biorthogonal polynomials (CBOPs) or not. The answer is affirmative.

In fact, let us consider the $t$-dependent monic CBOPs $\Phi(x;t)=(\Phi_0(x;t),\cdots,\Phi_{n-1}(x;t))^\mathrm{T}$ satisfying the biorthogonality
$$\iint_{\mathbb{R_+}^2}\frac{\Phi_k(x;t)\Phi_l(y;t)}{x+y}\mu(x;t)\mu(y;t)dxdy=h_k(t)\delta_{k,l},
$$
where 
$$\mu(x;t)dx=e^{tx}\sum_{p=1}^nc_p(0)\xi_p\delta(x-\xi_p)dx.
$$
In this case, it is clear that the $t$-dependent bimoments $$J_{i,j}(t)=\iint_{\mathbb{R_+}^2}\frac{x^iy^j}{x+y}\mu(x;t)\mu(y;t)dxdy$$ admit the evolution relation:
$$\dot J_{i,j}(t)=J_{i+1,j}+J_{i,j+1}.$$
As is shown, the CBOPs satisfy a four-term recurrence \cite{bertola2010cauchy,Miki2011cauchy}:
\begin{align}\label{ctoda_lax1}
xL_1\Phi(x;t)=L_2\Phi(x;t)
\end{align}
and one may obtain the evolution relation \cite{li2017cauchy}:
\begin{align}\label{ctoda_lax2}
L_1\dot \Phi(x;t) =B_2\Phi(x;t),
\end{align}
where 
\begin{eqnarray*}
&L_1=\left(\begin{array}{cccccc}
1&&\\
\mathcal{A}_1&1&\\
&\ddots&\ddots\\
&&\mathcal{A}_{n-1}&1
\end{array}
\right),\quad L_2=\left(\begin{array}{cccccc}
\mathcal{B}_0&1&&&\\
\mathcal{C}_1&\mathcal{B}_1&1&&\\
\mathcal{D}_2&\mathcal{C}_2&\mathcal{B}_2&1\\
&\ddots&\ddots&\ddots&\ddots\\
&&\mathcal{D}_{n-2}&\mathcal{C}_{n-2}&\mathcal{B}_{n-2}&1\\
&&&\mathcal{D}_{n-1}&\mathcal{C}_{n-1}&\mathcal{B}_{n-1}
\end{array}
\right),\\
&
B_2=\left(\begin{array}{cccccc}
0&&&\\
2(\mathcal{B}_0-\mathcal{A}_0)\mathcal{A}_1&0&&\\
&\ddots&\ddots\\
&&2(\mathcal{B}_{n-2}-\mathcal{A}_{n-2})\mathcal{A}_{n-1}&0
\end{array}
\right),
\end{eqnarray*}
with $$\mathcal{A}_k=\sqrt{\frac{b_ku_k}{b_{k-1}}},\quad \mathcal{B}_k=\frac{1}{2}b_k+\sqrt{\frac{b_ku_k}{b_{k-1}}},\quad \mathcal{C}_k=-u_k-\frac{1}{2}\sqrt{b_ku_kb_{k-1}},\quad \mathcal{D}_k=-u_{k-1}\sqrt{\frac{b_ku_k}{b_{k-1}}}.$$
\eqref{ctoda_lax1} and \eqref{ctoda_lax2} constitute the Lax pair of the finite C-Toda lattice \eqref{eq:ctoda}, in other words, the finite C-Toda lattice admit the Lax representation:
$$\dot L=[B,L],$$
where 
$$L=L_1^{-1}L_2,\qquad B=L_1^{-1}B_2.$$

\begin{remark}
The appearance of the evolution relation \eqref{ctoda_lax2} is slightly different from that in \cite{li2017cauchy}. The formulae here are neat and helpful to calculate the compatibility condition. By a direct computation, one can see that the above Lax formula is exactly the finite C-Toda lattice \eqref{eq:ctoda}. 
\end{remark}

\begin{remark}
An semi-infinite C-Toda lattice can be easily constructed together with its infinite matrix Lax formulation. One investigation can be found in \cite{li2017cauchy}. 
\end{remark}

\subsection{A correspondence between DP peakon and finite C-Toda lattices}
We have derived the nonlinear finite C-Toda lattice \eqref{eq:ctoda} from a linearised flow in the positive direction step by step. Recall that the DP peakon ODEs \eqref{DP_eq:peakon} can be linearised to a flow in the negative direction by inverse spectral method.  In Appendix \ref{app:CHA} and  \ref{app:NVB}, the relevance of the nonlinear systems between the CH peakon and the finite A-Toda lattices, and the Novikov peakon and the finite B-Toda lattices are summarised, respectively.
Similar to those in Appendix \ref{app:CHA} and   \ref{app:NVB},  we shall give an interpretation on the relation between nonlinear variables of DP peakon ODEs and the finite C-Toda lattice, which is summarised as the following theorem. 

\begin{theorem}\label{th:dp_ctoda} Given 
$$
\beta_i(0)=\sum_{p=1}^n\xi_p(0)^ic_p(0),\qquad J_{i,j}(0)=J_{j,i}(0)=\sum_{p=1}^n\sum_{q=1}^n\frac{\xi_p(0)^i\xi_q(0)^j}{\xi_p(0)+\xi_q(0)}c_p(0)c_q(0),
$$
with
 \[
 0<\xi_1(0)<\xi_2(0)<\cdots<\xi_n(0),\qquad c_p(0)>0,
 \]
let $\tau_{k}^{(i,j)}(0) $ and $\sigma_k^{(i,j)}(0)$ be defined as\footnote{The positivity of such determinants for $1\leq k\leq n$ is guaranteed (see e.g. \cite{bertola2010cauchy,lundmark2005degasperis}).}

\begin{align*}
\tau_k^{(i,j)}(0)=\left|
\begin{array}{cccc}
J_{i,j}(0)&J_{i,j+1}(0)&\cdots & J_{i,j+k-1}(0)\\
J_{i+1,j}(0)&J_{i+1,j+1}(0)&\cdots&J_{i+1,j+k-1}(0)\\
\vdots&\vdots&\ddots&\vdots\\
J_{i+k-1,j}(0)&J_{i+k-1,j+1}(0)&\cdots&J_{i+k-1,j+k-1}(0)
\end{array}
\right|=\tau_k^{(j,i)}(0)
\end{align*}
for any nonnegative integer $k$, and $\tau_0^{(i,j)}(0)=1$ , $\tau_k^{(i,j)}(0)=0$ for $k<0$, and \begin{align}
\sigma_k^{(i,j)}(0)=\left|
\begin{array}{ccccc}
\beta_{i}(0)&J_{i,j}(0)&J_{i,j+1}(0)&\cdots & J_{i,j+k-2}(0)\\
\beta_{i+1}(0)&J_{i+1,j}(0)&J_{i+1,j+1}(0)&\cdots&J_{i+1,j+k-2}(0)\\
\vdots&\vdots&\vdots&\ddots&\vdots\\
\beta_{i+k-1}(0)&J_{i+k-1,j}(0)&J_{i+k-1,j+1}(0)&\cdots&J_{i+k-1,j+k-2}(0)
\end{array}
\right|
\end{align}
with the convention $\sigma_1^{(i,j)}(0)=\beta_{i}(0)$ and $\sigma_k^{(i,j)}(0)=0$ for $k<1$.
 \begin{enumerate}
 \item Introduce the variables $\{x_k(0),m_k(0)\}_{k=1}^n$ defined by
\begin{align*}
 x_{k'}(0)=\frac{1}{2}\log\frac{2\tau_k^{(0,1)}(0)}{\tau_{k-1}^{(1,2)}(0)},\qquad m_{k'}(0)=\frac{\tau_k^{(0,1)}(0)\tau_{k-1}^{(1,2)}(0)}{\tau_k^{(1,1)}(0)\tau_{k-1}^{(1,1)}(0)},
\end{align*}
where $k'=n+1-k.$ If $\{\xi_p(t),c_p(t)\}_{p=1}^n$ evolve as 
 \[
 \dot \xi_p=0,\qquad \dot c_p=\frac{c_p}{\xi_p},
 \]
then $\{x_k(t),m_k(t)\}_{k=1}^n$  satisfy the DP peakon ODEs \eqref{DP_eq:peakon}.
\item 
 Introduce the variables $\left\{\{u_k(0)\}_{k=1}^{n-1},\{b_{k}(0)\}_{k=0}^{n-1}\right\}$ defined by
$$u_k(0)=\frac{\tau_{k+1}^{(1,1)}(0)\tau_{k-1}^{(1,1)}(0)}{({\tau_{k}^{(1,1)}(0)})^2},\qquad b_k=\frac{(\sigma_{k+1}^{(1,1)}(0))^2}{\tau_{k+1}^{(1,1)}(0)\tau_{k}^{(1,1)}(0)}.$$
 If $\{\xi_p(t),c_p(t)\}_{p=1}^n$ evolve as 
 \[
 \dot \xi_p=0,\qquad \dot c_p=\xi_pc_p,
 \]
 then  $\left\{\{u_k(0)\}_{k=1}^{n-1},\{b_{k}(0)\}_{k=0}^{n-1}\right\}$ satisfy the finite C-Toda lattice \eqref{eq:ctoda} with $u_0=u_n=0$.
 \item There exists a mapping from $\{x_k(0),m_k(0)\}_{k=1}^n$ to $\left\{\{u_k(0)\}_{k=1}^{n-1},\{b_{k}(0)\}_{k=0}^{n-1}\right\}$ according to
\begin{align*}
 &u_k(0)=\frac{e^{2x_{k'-1}(0)-2x_{k'}(0)}}{m_{k'-1}(0)m_{k'}(0)\left(1-e^{x_{k'-1}(0)-x_{k'}(0)}\right)^4}, \\
 & b_k(0)=\frac{2(e^{x_{k'}(0)-x_{k'-1}(0)}-e^{x_{k'-2}(0)-x_{k'-1}(0)})^2}{m_{k'-1}(0)(1-e^{x_{k'-2}(0)-x_{k'-1}(0)})^2(e^{x_{k'}(0)-x_{k'-1}(0)}-1)^2}
\end{align*}
with the convention
\begin{align*}
 x_0(0)=-\infty,\quad x_{n+1}(0)=+\infty.
\end{align*}
Actually, this item partially reveals the relevance of the corresponding spectral problems of the DP peakon and C-Toda lattices.
 \end{enumerate}
 \end{theorem}
 \begin{proof}
The first two results immediately follow as they are the subjects of Theorem \ref{th:DP_solution} and Theorem \ref{th:ctoda}.

For the third one, we need two observations, one of which is
 $$\tau_k^{(1,1)}(0)=\sqrt{\tau_k^{(0,1)}(0)\tau_k^{(1,2)}(0)}-\sqrt{\tau_{k+1}^{(0,1)}(0)\tau_{k-1}^{(1,2)}(0)}$$ 
 obtained from the corresponding formulae for $\tau_k^{(i,j)}(0)$, $\sigma_k^{(i,j)}(0)$ in terms of $U_k,V_k$ formally in Lemma \ref{lem:FGUVW}. 
 The other is
   $$\sigma_k^{(1,1)}(0)=\sqrt{\tau_{k-1}^{(0,1)}(0)\tau_k^{(1,2)}(0)}-\sqrt{\tau_{k+1}^{(0,1)}(0)\tau_{k-2}^{(1,2)}(0)},$$ 
which can be derived by noticing 
  \begin{align}\label{sigma1120}
  \sigma_k^{(1,1)}(0)=(-1)^{k-1}\left|
\begin{array}{ccccc}
\beta_{1}(0)&J_{2,0}(0)&J_{2,1}(0)&\cdots & J_{2,k-2}(0)\\
\beta_{2}(0)&J_{3,0}(0)&J_{3,1}(0)&\cdots&J_{3,k-2}(0)\\
\vdots&\vdots&\vdots&\ddots&\vdots\\
\beta_{k}(0)&J_{k+1,0}(0)&J_{k+1,1}(0)&\cdots&J_{k+1,k-2}(0)
\end{array}
\right|,
\end{align}
and the corresponding formulae  in Lemma \ref{lem:FGUVW}. Here \eqref{sigma1120} is a consequence of row and column operations by using the fact
$$J_{i+1,j}+J_{i,j+1}=\beta_{i}\beta_{j}.$$ 
Then the proof for the third one can be completed by direct calculations.
\begin{align*}
\end{align*}
\end{proof}

\section{Relation between C-Toda lattice and discrete CKP equation}\label{sec:schief}
In this section, we explain how to connect the C-Toda lattice with the discrete CKP equation \cite{bobenko2015discrete,bobenko2017discrete,fu17,schief2003lattice}. This provides a reasonable argument to the origin of ``the Toda lattice of CKP type''. 

The discrete CKP equation 
\begin{align}\label{eq:dckp}
&\left(\hat\tau_{i,j,l}\hat\tau_{i+1,j+1,l+1}+\hat\tau_{i+1,j,l}\hat\tau_{i,j+1,l+1}-\hat\tau_{i,j+1,l}\hat\tau_{i+1,j,l+1}-\hat\tau_{i,j,l+1}\hat\tau_{i+1,j+1,l}\right)^2\nonumber\\
=&4\left(\hat\tau_{i+1,j+1,l}\hat\tau_{i+1,j,l+1}-\hat\tau_{i+1,j,l}\hat\tau_{i+1,j+1,l+1}\right)\left(\hat\tau_{i,j+1,l}\hat\tau_{i,j,l+1}-\hat\tau_{i,j,l}\hat\tau_{i,j+1,l+1}\right)
\end{align}
was derived by Schief as an integrable system from a geometric constraint of the discrete Darboux system associated with a theorem due to Carnot \cite{schief2003lattice}.  It has also been obtained by Kashaev \cite{kashaev1996discrete} in the context of star triangle moves in the Ising model. It can also be regarded as a reduction of the hexahedron recurrence in connection with cluster algebras and dimer configurations \cite{kenyon2016double}.  Furthermore, it is shown that the discrete CKP equation admits a tau-function generated by a symmetric $M$-system, and can be interpreted as a local relation between the principal minors of the $M$-system \cite{bobenko2015discrete,bobenko2017discrete,holtz2007hyper}. All these facts show that the discrete CKP equation possesses interesting algebraic and geometric properties.

Under the coordinate transformation
$$
\hat\tau_{i,j,l}=\tilde\tau_{k,m,n},\ \qquad\qquad k=i-j,\ m=l, \ n=j,
$$
the discrete CKP equation \eqref{eq:dckp} can be rewritten as 
\begin{align}\label{eq:dckp2}
&\left(\tilde\tau_{k,m,n}\tilde\tau_{k,m+1,n+1}+\tilde\tau_{k+1,m,n}\tilde\tau_{k-1,m+1,n+1}-\tilde\tau_{k+1,m+1,n}\tilde\tau_{k-1,m,n+1}-\tilde\tau_{k,m,n+1}\tilde\tau_{k,m+1,n}\right)^2\nonumber\\
=&4\left(\tilde\tau_{k+1,m+1,n}\tilde\tau_{k,m,n+1}-\tilde\tau_{k+1,m,n}\tilde\tau_{k,m+1,n+1}\right)\left(\tilde\tau_{k,m+1,n}\tilde\tau_{k-1,m,n+1}-\tilde\tau_{k,m,n}\tilde\tau_{k-1,m+1,n+1}\right)
\end{align}
Then applying a 1+1 dimensional reduction with the index constraint $\tilde\tau_{k,m,n}\mapsto\tilde\tau_{k,m+n}$, we have 
\begin{align}\label{eq:1+1dckp}
&\left(\tilde\tau_{k,m}\tilde\tau_{k,m+2}+\tilde\tau_{k+1,m}\tilde\tau_{k-1,m+2}-\tilde\tau_{k+1,m+1}\tilde\tau_{k-1,m+1}-\tilde\tau_{k,m+1}\tilde\tau_{k,m+1}\right)^2\nonumber\\
=&4\left(\tilde\tau_{k+1,m+1}\tilde\tau_{k,m+1}-\tilde\tau_{k+1,m}\tilde\tau_{k,m+2}\right)\left(\tilde\tau_{k,m+1}\tilde\tau_{k-1,m+1}-\tilde\tau_{k,m}\tilde\tau_{k-1,m+2}\right).
\end{align}
which we refer to as a 1+1 dimensional reduction of the discrete CKP equation.

Now we are ready to state the relation between the C-Toda lattice and the discrete CKP equation, which is summarised in the following theorem.
\begin{theorem}
As for the 1+1 dimensional reduction of the discrete CKP equation \eqref{eq:1+1dckp}, under the change of variables
 $$\tau_k(t)=\epsilon^{-k^2}\tilde\tau_{k,m},\qquad t=m\epsilon,$$
one can arrive at the C-Toda lattice of quartic form \eqref{eq:ctoda_quartic} as $m\rightarrow\infty,\epsilon\rightarrow 0$.
\end{theorem}
\begin{proof}
Under the variable transformation, we have
$$
\tilde\tau_{k,m+1}=\epsilon^{k^2}\tau_k(t+\epsilon), \qquad \tilde\tau_{k,m-1}=\epsilon^{k^2}\tau_k(t-\epsilon),\qquad \tilde\tau_{k,m+2}=\epsilon^{k^2}\tau_k(t+2\epsilon).
$$
Substituting Taylor's expansion 
\begin{align*}
&\tilde\tau_{k,m+1}=\epsilon^{k^2}(\tau_k(t)+\epsilon\dot\tau_k(t)+\frac{1}{2}\epsilon^2\ddot\tau_k(t)+\cdots),\\
&\tilde\tau_{k,m-1}=\epsilon^{k^2}(\tau_k(t)-\epsilon\dot\tau_k(t)+\frac{1}{2}\epsilon^2\ddot\tau_k(t)+\cdots),\\
&\tilde\tau_{k,m+2}=\epsilon^{k^2}(\tau_k(t)+2\epsilon\dot\tau_k(t)+2\epsilon^2\ddot\tau_k(t)+\cdots)
\end{align*}
into \eqref{eq:1+1dckp} and taking the limit  $\epsilon\rightarrow 0$, one can obtain the C-Toda lattice of quartic form \eqref{eq:ctoda_quartic}. 
\end{proof}
\begin{remark}
We have already expressed the tau function of the finite C-Toda lattice in terms of Cauchy bimoment determinants in Section \ref{sec:ctoda}, and associated it with Cauchy biorthogonal polynomials. Actually we can also obtain some parallel results for the fully discrete CKP equation \eqref{eq:dckp2} and the 1+1 dimensional reduction \eqref{eq:1+1dckp}. However, this will be the subject of a future publication.
\end{remark}


\section{Conclusion and discussion}\label{sec:con}
In this paper, we derived a finite C-Toda lattice together with Lax representation and showed that it can be regarded as a nonlinear opposite flow to the DP peakon lattice. So far, it has turned out that the CH, Novikov, DP peakon lattices are related to A,B,C-Toda lattices, respectively. 



It is known that, the CH, Novikov and DP equations can, via hodograph transformations, be connected with the first negative flows of KdV, SK and KK hierarchies, respectively \cite{degasperis2002new,fuchssteiner1996some,hone2008integrable}. And, the KdV, SK and KK hierarchies are in the frame of AKP, BKP, CKP hierarchies, which are classified according to the associated Lie algebra of the respective tau function structures \cite{jimbo1983solitons}. The CH peakon and A-Toda, Novikov peakon and B-Toda, DP peakon and C-Toda lattices share the same tau function structures, and  form opposite nonlinear flows, respectively. To summarise, the relations may be charactered by the following table.

\begin{table}[htbp]
\caption{A comparison between CH, Novikov, DP peakon and A,B,C-Toda lattices }
\begin{tabular}{|c|c|c|c|c|c|}
\hline
\multicolumn{2}{|c|}{\multirow{2}{*}{Hierarchy}}&\multirow{2}{*}{$\tau$-structure}&\multirow{2}{*}{Moments $c_{i,j}$}  &
\multicolumn{2}{c|}{ $t$-deformation flow}\\
\cline{5-6}
\multicolumn{2}{|c|}{}&&&
$\mu (x;t)=e^{tx}\mu(x;0)$&$\mu (x;t)=e^{\frac{t}{x}}\mu(x;0)$\\ \hline
KP&KdV&$\det((c_{i,j}))$&$\int x^{i+j} \mu (x)dx $&A-Toda& CH peakon \\
BKP&SK&$Pf((c_{i,j}))$&$\iint \frac{x-y}{x+y}x^iy^j \mu (x)\mu (y)dxdy$&B-Toda&Novikov peakon\\
CKP&KK&$\det((c_{i,j}))$&$\iint \frac{x^iy^j}{x+y}\mu (x)\mu (y) dxdy$&C-Toda &DP peakon\\
\hline
\end{tabular}
\label{comp_abc}
\end{table}

 It is well known that the A-Toda lattice has been extensively studied in the literature. The B-Toda lattice, especially the bilinear form, has also been discussed in some papers, e.g. \cite{chang2017partial,hirota2001soliton,hu2017partition}. The C-Toda lattice derived in the present paper looks novel and concise and remains further investigation. 
 

\section{Acknowledgements}
We thank anonymous referees for many helpful suggestions. This work was supported in part by the National Natural Science Foundation of China. X.C. was supported under the grant nos 11701550, 11731014. X.H. was supported under the grant nos 11331008, 11571358.

%
%
%

\begin{appendix}
\section{CH peakon and finite A-Toda lattices}\label{app:CHA}

The finite A-Toda lattice reads
\begin{subequations}
\label{eq:atoda}
\begin{align}
&\dot u_{k}=u_{k}(b_{k}-b_{k-1}),  \quad \ k=1,\dots,n-1,\\
&\dot b_k=u_{k+1}-u_k, \qquad\quad k=0,\dots,n-1,
\end{align}
\end{subequations}
with the boundary condition $u_0=u_n=0.$ It can be equivalently written as the Lax equation
$$
\dot L=[L,B]
$$
with
\begin{align*}
L=\left(
\begin{array}{ccccc}
b_0&1\\
u_1&b_1&1\\
&\ddots&\ddots&\ddots\\
&&u_{n-2}&b_{n-2}&1\\
&&&u_{n-1}&b_{n-1}
\end{array}
\right),\qquad B=\left(
\begin{array}{ccccc}
0&0\\
u_1&0&0\\
&\ddots&\ddots&\ddots\\
&&u_{n-2}&0&0\\
&&&u_{n-1}&0
\end{array}
\right).
\end{align*} 
It is not hard to see that the ordinary orthogonal polynomials appear as wave functions of the Lax pair of the finite A-Toda lattice. Indeed,
 the finite A-Toda lattice can arise as a one-parameter deformation of the measure of the ordinary orthogonal polynomials.

The relation between the CH peakon and the finite A-Toda lattices was revealed in \cite{beals2001peakons}, where a mapping is established via the corresponding spectral problems. In this appendix, in order to put it in a unified frame, we give a slightly different connection from that in \cite{beals2001peakons}. 
\begin{theorem}\label{th:ch_atoda} Given 
$$
A_i(0)=\sum_{p=1}^n\xi_p^i(0)c_p(0),\qquad \text{with}\qquad  0<\xi_1(0)<\xi_2(0)<\cdots<\xi_n(0),\qquad c_p(0)>0,
$$
let $\tau_k^{(l)}(0)$ be defined as Hankel determinants\footnote{The positivity for such determinants with $1\leq k\leq n$ is guaranteed (see e.g. \cite{beals2000multipeakons,chang2014generalized}).}
\begin{align*}
\tau_k^{(l)}(0)=\left|
\begin{array}{cccc}
A_l(0)&A_{l+1}(0)&\cdots &A_{l+k-1}(0)\\
A_{l+1}(0)&A_{l+2}(0)&\cdots &A_{l+k}(0)\\
\vdots&\vdots&\ddots&\vdots\\
A_{l+k-1}(0)&A_{l+k}(0)&\cdots &A_{l+2k-2}(0)\\
\end{array}
\right|
\end{align*}
for any nonnegative integer $k$, and $\tau_0^{(l)}(0)=1$ , $\tau_k^{(l)}(0)=0$ for $k<0$.
 \begin{enumerate}
 \item Introduce the variables $\{x_k(0),m_k(0)\}_{k=1}^n$ defined by
\begin{align*}
 x_{k'}(0)=\ln \left(\frac{2\tau_k^{(0)}(0)}{\tau_{k-1}^{(2)}(0)}\right),\qquad m_{k'}(0)=\frac{\tau_k^{(0)}(0)\tau_{k-1}^{(2)}(0)}{\tau_k^{(1)}(0)\tau_{k-1}^{(1)}(0)},
\end{align*}
where $k'=n+1-k.$ If $\{\xi_p(t),c_p(t)\}_{p=1}^n$ evolve as 
 \[
 \dot \xi_p=0,\qquad \dot c_p=\frac{c_p}{\xi_p},
 \]
then $\{x_k(t),m_k(t)\}_{k=1}^n$  satisfy the CH peakon ODEs:
\begin{align*}
\dot x_k=\sum_{j=1}^nm_je^{-|x_j-x_k|},\qquad \dot m_k=\sum_{j=1}^n\sgn(x_k-x_j)m_je^{-|x_j-x_k|}.
\end{align*}
\item 
 Introduce the variables $\left\{\{u_k(0)\}_{k=1}^{n-1},\{b_{k}(0)\}_{k=0}^{n-1}\right\}$ defined by
$$u_k(0)=\frac{\tau_{k+1}^{(1)}(0)\tau_{k-1}^{(1)}(0)}{({\tau_{k}^{(1)}(0)})^2},\qquad b_k=\frac{\tau_{k+1}^{(1)}(0)\tau_{k-1}^{(2)}(0)}{\tau_k^{(2)}(0)\tau_{k}^{(1)}(0)}+\frac{\tau_{k+1}^{(2)}(0)\tau_{k}^{(1)}(0)}{\tau_k^{(2)}(0)\tau_{k+1}^{(1)}(0)}.$$
 If $\{\xi_p(t),c_p(t)\}_{p=1}^n$ evolve as 
 \[
 \dot \xi_p=0,\qquad \dot c_p=\xi_pc_p,
 \]
 then $\left\{\{u_k(0)\}_{k=1}^{n-1},\{b_{k}(0)\}_{k=0}^{n-1}\right\}$  satisfy the finite A-Toda lattice \eqref{eq:atoda} with $u_0=u_n=0$.
 \item There exists a mapping from $\{x_k(0),m_k(0)\}_{k=1}^n$ to $\left\{\{u_k(0)\}_{k=1}^{n-1},\{b_{k}(0)\}_{k=0}^{n-1}\right\}$ according to
\begin{align*}
 &u_k(0)=\frac{e^{x_{k'-1}(0)-x_{k'}(0)}}{m_{k'-1}(0)m_{k'}(0)\left(1-e^{x_{k'-1}(0)-x_{k'}(0)}\right)^2}, \\
 & b_k(0)=\frac{1}{m_{k'-1}(0)(1+e^{-x_{k'-1}(0)})}\left(\frac{1+e^{-x_{k'}(0)}}{1-e^{x_{k'-1}(0)-x_{k'}(0)}}+\frac{1+e^{-x_{k'-1}(0)}}{1-e^{x_{k'-2}(0)-x_{k'-1}(0)}}\right)
 \end{align*}
with the convention
\begin{align*}
x_0(0)=-\infty,\quad x_{n+1}(0)=+\infty.
\end{align*}
Actually, this item partially reveals the relevance of the corresponding spectral problems of the CH peakon and A-Toda lattices.
 \end{enumerate}
 \end{theorem}
\begin{proof}
The first item is the main content of the explicit construction of CH peakon ODEs in \cite{beals2000multipeakons} with certain scalling. The second one is on the explicit formulae of the finite A-Toda lattice \cite{beals2001peakons}. The third one can be directly confirmed by using an equality
$$\tau_{k+1}^{(0)}\tau_{k-1}^{(2)}=\tau_{k}^{(0)}\tau_{k}^{(2)}-\left(\tau_{k}^{(1)}\right)^2$$
obtained by the Jacobi's identity \eqref{id:jacobi}.
\end{proof}

It follows from the above theorem that A-Toda and CH peakon lattices can be seen as opposite nonlinear flows to each other.

\section{Novikov peakon and finite B-Toda lattices}\label{app:NVB}
The correspondence of the Novikov peakon and finite B-Toda lattices was investigated in the recent work \cite{chang2017application}. This appendix is devoted to listing the main bridge between them with a slight modification.

The finite B-Toda lattice reads
\begin{subequations}
\label{eq:btoda}
\begin{align}
&\dot u_{k}=u_{k}(b_{k}-b_{k-1}),  \qquad\qquad\qquad\qquad\ \  \ k=1,\dots,n-1,\\
&\dot b_k=u_{k+1}(b_{k+1}+b_k)-u_{k}(b_{k}+b_{k-1}), \quad k=0,\dots,n-1,
\end{align}
\end{subequations}
with the boundary condition $u_0=u_n=0.$  The Lax representation of the finite B-Toda lattice was constructed in \cite{chang2017partial}. In fact, an appropriate deformation of the so-called partial-skew orthogonal polynomials (PSOPs) may lead to the finite B-Toda lattice, that is, the finite B-Toda lattice can be equivalently written as
$$\dot L=[B,L],$$
where 
\begin{eqnarray*}
&L=L_1^{-1}L_2,\qquad B=L_1^{-1}B_2,\\
&L_1=\left(\begin{array}{cccccc}
1&&\\
u_1&1&\\
&\ddots&\ddots\\
&&u_{n-1}&1
\end{array}
\right),\quad
B_2=\left(\begin{array}{cccccc}
0&&&\\
\mathcal{A}_1&0&&\\
&\ddots&\ddots\\
&&\mathcal{A}_{n-1}&0
\end{array}
\right),\\
&L_2=\left(\begin{array}{cccccc}
\mathcal{B}_0&1&&&\\
\mathcal{C}_1&\mathcal{B}_1&1&&\\
\mathcal{D}_2&\mathcal{C}_2&\mathcal{B}_2&1\\
&\ddots&\ddots&\ddots&\ddots\\
&&\mathcal{D}_{n-2}&\mathcal{C}_{n-2}&\mathcal{B}_{n-2}&1\\
&&&\mathcal{D}_{n-1}&\mathcal{C}_{n-1}&\mathcal{B}_{n-1}
\end{array}
\right)
\end{eqnarray*}
with $$\mathcal{A}_k=u_k(b_k+b_{k-1}),\quad \mathcal{B}_k=u_k+b_{k},\quad \mathcal{C}_k=-u_k(b_k+u_{k+1}),\quad \mathcal{D}_k=-(u_k)^2u_{k-1}.$$

The following conclusion is a consequence of combining Th. 3.7, 3.9 and 3.10 in \cite{chang2017application},  which implies that it is reasonable to view the B-Toda lattice and Novikov peakon lattice as opposite flows. Just note that here we establish a straightforward connection between the B-Toda lattice (instead of the modified B-Toda lattice) and Novikov peakon lattice. Actually it is not hard to confirm that there exists a Miura transformation from the modified B-Toda lattice to the B-Toda lattice. In order to give a unified picture in the present paper, it is convenient and sufficient for us to make use of the B-Toda lattice.
\begin{theorem}\label{th:nv_btoda}
Given Pfaffian entries
$$
Pf(i,j)|_{{t=0}}=\sum_{p=1}^n\sum_{q=1}^n
\frac{\zeta_p(0)-\zeta_q(0)}{\zeta_p(0)+\zeta_q(0)}(\zeta_p(0))^i(\zeta_q(0))^jc_p(0)c_q(0),\qquad Pf(d_0,i)|_{{t=0}}=\sum_{p=1}^nc_p(0)(\zeta_p(0))^i$$
with
 \[
 0<\zeta_1(0)<\zeta_2(0)<\cdots<\zeta_n(0),\qquad c_p(0)>0,
 \]
let $\tau_k^{(l)}(0)$ be defined as Pfaffians\footnote{The positivity of $\tau_k^{(l)}$ and $W_k^{(l)}$ for $1\leq k\leq n$ is ensured (see e.g. \cite{chang2017application,hone2009explicit}).
}
$$\tau_{2k}^{(l)}(0)=Pf(l,l+1,\cdots,l+2k-1)|_{{t=0}},\quad \tau_{2k-1}^{(l)}(0)=Pf(d_0,l,l+1,\cdots,l+2k-2)|_{{t=0}}$$
for any nonnegative integers $k$, and $\tau_0^{(l)}(0)=1$ , $\tau_k^{(l)}(0)=0$ for $k<0$.
Define
$$W_k^{(l)}(0)=\tau_k^{(l+1)}(0)\tau_k^{(l)}(0)-\tau_{k-1}^{(l+1)}(0)\tau_{k+1}^{(l)}(0).$$
 \begin{enumerate}
 \item Introduce the variables $\{x_k(0),m_k(0)\}_{k=1}^n$ defined by
$$ x_{k'}(0)=\frac{1}{2}\ln\frac{W_k^{(-1)}(0)}{W_{k-1}^{(0)}(0)},\qquad m_{k'}(0)=\frac{\sqrt{W_k^{(-1)}(0)W_{k-1}^{(0)}(0)}}{\tau_k^{(0)}(0)\tau_{k-1}^{(0)}(0),}$$
where $k'=n+1-k.$ If $\{\zeta_p(t),c_p(t)\}_{j=1}^n$ evolve as 
 \[
 \dot \zeta_p=0,\qquad \dot c_p=\frac{c_p}{\zeta_p},
 \]
then $\{x_k(t),m_k(t)\}_{k=1}^n$  satisfy the Novikov peakon ODEs:
\begin{align*}
&\dot x_k=\left(\sum_{j=1}^nm_je^{-|x_j-x_k|}\right)^2,\\
&\dot m_k=\left(\sum_{j=1}^nm_je^{-|x_j-x_k|}\right)\left(\sum_{j=1}^n\sgn(x_k-x_j)m_je^{-|x_j-x_k|}\right).
\end{align*}
\item 
 Introduce the variables $\{\{u_k(0)\}_{k=1}^{n-1},\{b_{k}(0)\}_{k=0}^{n-1}\}$ defined by
\begin{align*}
&u_k(0)=\frac{\tau_{k+1}^{(0)}(0)\tau_{k-1}^{(0)}(0)}{({\tau_{k}^{(0)}(0)})^2},\\
& b_k(0)=\frac{\tau_{k+1}^{(0)}(0)\tau_{k-1}^{(0)}(0)}{({\tau_{k}^{(0)}(0)})^2}+\frac{\tau_{k+2}^{(0)}(0)\tau_{k}^{(0)}(0)}{({\tau_{k+1}^{(0)}(0)})^2}+\frac{\left(\tau_{k+1}^{(0)}(0)\right)^2W_{k-1}^{(0)}(0)}{\left(\tau_{k}^{(0)}(0)\right)^2W_{k}^{(0)}(0)}+\frac{\left(\tau_{k}^{(0)}(0)\right)^2W_{k+1}^{(0)}(0)}{\left(\tau_{k+1}^{(0)}(0)\right)^2W_{k}^{(0)}(0)}.
\end{align*}
 If $\{\zeta_p(t),c_p(t)\}_{j=1}^n$ evolve as 
 \[
 \dot \zeta_p=0,\qquad \dot c_p=\zeta_pc_p,
 \]
 then  $\{\{u_k(0)\}_{k=1}^{n-1},\{b_{k}(0)\}_{k=0}^{n-1}\}$  satisfy the finite B-Toda lattice \eqref{eq:btoda} with $u_0=u_n=0$.
 \item There exists a mapping from $\{x_k(0),m_k(0)\}_{k=1}^n$ to  $\left\{\{u_k(0)\}_{k=1}^{n-1},\{b_{k}(0)\}_{k=0}^{n-1}\right\}$ according to
 \begin{align*}
 &u_k(0)=\frac{1}{2m_{k'}(0)m_{k'-1}(0)\cosh x_{k'}(0)\cosh x_{k'-1}(0)(\tanh x_{k'}(0)-\tanh x_{k'-1}(0))}, \\
 &b_k(0)=u_k(0)\left(1+\frac{m_{k'}(0)}{m_{k'-1}(0)}e^{x_{k'-1}(0)-x_{k'}(0)}\right)+u_{k+1}(0)\left(1+\frac{m_{k'-2}(0)}{m_{k'-1}(0)}e^{x_{k'-1}(0)-x_{k'-2}(0)}\right)
\end{align*}
with the convention
$$ x_0(0)=-\infty,\quad x_{n+1}(0)=+\infty.$$
Actually, this item partially reveals the relevance of the corresponding spectral problems of the Novikov peakon and B-Toda lattices.

 \end{enumerate}
 \end{theorem}
 \begin{proof}
The first item is the subject of the multipeakon formulae of the Novikov equation in \cite{hone2009explicit}. The second one can be proved by noticing that in the B-Toda lattice, $b_k$ admits the solution of the form 
$$b_k=\frac{d}{dt}\log\tau_{k+1}^{(0)}- \frac{d}{dt}\log\tau_{k}^{(0)},$$
which can yield the required formulae by using the modified B-Toda lattice together with its solution in \cite{chang2017application}. For the third one, it works by a direct calculation. 
 \end{proof}



\end{appendix}
    
\small
\bibliographystyle{abbrv}
\bibliographystyle{plain}
%

\def\cydot{\leavevmode\raise.4ex\hbox{.}}
  \def\cydot{\leavevmode\raise.4ex\hbox{.}} \def\cprime{$'$}

\end{document}